\documentclass[sigconf]{acmart}

\copyrightyear{2022}
\acmYear{2022}
\setcopyright{rightsretained}

\acmConference[SPAA '22]{Proceedings of the 34th ACM Symposium on Parallelism in Algorithms and Architectures}{July 11--14, 2022}{Philadelphia, PA, USA}
\acmBooktitle{Proceedings of the 34th ACM Symposium on Parallelism in Algorithms and Architectures (SPAA '22), July 11--14, 2022, Philadelphia, PA, USA}
\acmDOI{10.1145/3490148.3538572}
\acmISBN{978-1-4503-9146-7/22/07}

\usepackage{etoolbox}
\makeatletter
\patchcmd{\maketitle}{\@copyrightpermission}{
   \begin{minipage}{0.3\columnwidth}
     \href{https://creativecommons.org/licenses/by/4.0/}{\includegraphics[width=0.90\textwidth]{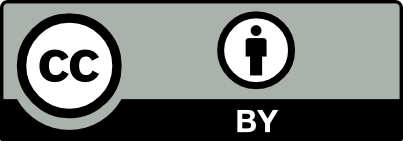}}
   \end{minipage}\hfill
   \begin{minipage}{0.7\columnwidth}
     \href{https://creativecommons.org/licenses/by/4.0/}{This work is licensed under a Creative Commons Attribution International 4.0 License.}
   \end{minipage}

   \vspace{5pt}
}{}{}

\makeatother

\usepackage[figure,algo2e,linesnumbered,vlined]{algorithm2e}
\usepackage{algorithmicx}
\usepackage{algorithm}
\usepackage[noend]{algpseudocode}
\usepackage{graphicx}
\usepackage{varwidth}
\usepackage{xspace}
\usepackage{url}
\usepackage{subcaption}
\usepackage{subfloat}
\usepackage{multicol}
\usepackage{hyperref}
\usepackage{appendix}
\usepackage{placeins}
\usepackage{moresize}

\graphicspath{{./pictures/}}

\SetInd{0.2em}{0.4em}

\SetAlFnt{\scriptsize\sf}

\usepackage{mathtools}

\begin{document}

\title[wCQ: A Fast Wait-Free Queue with Bounded Memory Usage]{wCQ: A Fast Wait-Free Queue with Bounded Memory Usage}

\author{Ruslan Nikolaev}
\email{rnikola@psu.edu}
\affiliation{%
        \institution{The Pennsylvania State University}
        \city{University Park}
        \state{PA}
        \country{USA}
}

\author{Binoy Ravindran}
\email{binoy@vt.edu}
\affiliation{%
        \institution{Virginia Tech}
        \city{Blacksburg}
        \state{VA}
        \country{USA}
}

\begin{abstract}
The concurrency literature presents a number of approaches for building non-blocking, FIFO, multiple-producer and multi\-ple-consu\-mer (MPMC) queues. However, only a fraction of them have high performance. In addition, many queue designs, such as LCRQ, trade memory usage for better performance. The recently proposed SCQ design achieves both memory efficiency as well as excellent performance. Unfortunately, both LCRQ and SCQ are only lock-free. On the other hand, existing wait-free queues are either not very performant or suffer from potentially unbounded memory usage. Strictly described, the latter queues, such as Yang \& Mellor-Crummey's (YMC) queue, forfeit wait-freedom as they are blocking when memory is exhausted.

We present a wait-free queue, called wCQ. wCQ is based on SCQ and uses its own variation of fast-path-slow-path methodology to attain wait-freedom and bound memory usage. Our experimental studies on x86 and PowerPC architectures validate wCQ's great performance and memory efficiency. They also show that wCQ's performance is often on par with the best known concurrent queue designs.

\end{abstract}



\maketitle

\algnewcommand{\Null}{\textbf{null}}%
\algnewcommand{\Last}{\textbf{last}}
\algnewcommand{\algorithmicgoto}{\textbf{goto}}%
\algnewcommand{\Goto}[1]{\algorithmicgoto~\ref{#1}}%
\algdef{SE}[DOWHILE]{Do}{doWhile}{\algorithmicdo}[1]{\algorithmicwhile\ #1}%
\algnewcommand\Not{\textbf{!}}
\algnewcommand\AndOp{\textbf{and}\xspace}
\algnewcommand\ModOp{\textbf{mod}\xspace}
\algnewcommand\OrOp{\textbf{or}\xspace}

\SetKw{Break}{break}
\SetKw{Continue}{continue}
\SetKwIF{If}{ElseIf}{Else}{if (}{)}{else if}{else}{endif}

\SetKwRepeat{Do}{do}{while}
\SetKwProg{Fn}{}{}{}
\newcommand{\removelatexerror}{\let\@latex@error\@gobble}

\algnewcommand{\LineComment}[1]{\State \(\triangleright\) #1}

\section{Introduction}
\label{sec:intro}

The concurrency literature presents an array of efficient non-block\-ing data structures with various types of progress properties. \textit{Lock-free} data
structures, where \textit{some} thread must complete an operation after
a finite number of steps, have traditionally received substantial practical
attention. \textit{Wait-free} data structures, which provide even stronger
progress properties by requiring that \textit{all} threads complete any
operation after a finite number of steps, have been less popular
since they were harder to design and much slower than
their lock-free alternatives. Nonetheless, the design
methodologies have evolved over the years, and wait-free data structures have
increasingly gained more attention because of their strongest progress
property.

Wait-freedom is
attractive for a number of reasons: (1) lack of starvation and reduced tail latency; (2) security; (3) reliability. Even if we assume that the scheduler is non-adversarial, applications can still be buggy or malicious. When having shared memory between two entities or applications, aside from parameter verification,
bounding the number of operations in loops is desirable for security (DoS) and reliability reasons. Note that no theoretical bound exists for lock-free algorithms even if they generally have similar performance.

Creating efficient FIFO queues, let alone wait-free ones, is notoriously hard. Elimination techniques, so familiar for LIFO stacks~\cite{Hendler:2004:SLS:1007912.1007944}, are not that useful in the FIFO world. Although~\cite{Moir:2005:UEI:1073970.1074013} does describe FIFO elimination, it only works well for certain shorter queues. Additionally, it is not always possible to alter the FIFO discipline, as proposed in~\cite{Kirsch:2013:FSL:2960356.2960376,10.1145/3437801.3441583}.

FIFO queues are widely used in a variety of practical applications, which range from ordinary memory pools to programming languages and high-speed networking.
Furthermore, for building efficient user-space message passing
and scheduling, non-blocking algorithms are especially desirable
since they avoid mutexes and kernel-mode switches.

FIFO queues are instrumental for certain synchronization primitives. A number of languages, e.g., Vlang, Go, can benefit from having a fast queue for their concurrency and synchronization constructs. For example, Go needs a queue for its buffered channel implementation. Likewise, high-speed networking and storage libraries such as DPDK~\cite{DPDK} and SPDK~\cite{SPDK}
use \textit{ring buffers} (i.e., bounded circular queues) for various purposes when allocating
and transferring network frames, inter-process tracepoints, etc.
Oftentimes, straight-forward implementations of ring buffers (e.g., the one found in DPDK) are erroneously dubbed as ``lock-free'' or ``non-blocking'', whereas those
queues merely avoid \textit{explicit} locks but still lack true non-blocking
progress guarantees.

More specifically, such queues require a thread to reserve a
ring buffer slot prior to writing new data. These approaches, as previously
discussed~\cite{Feldman:2015:WMM:2835260.2835264,ringdisapp}, are technically
blocking
since one stalled (e.g., preempted) thread in the middle of an operation can
adversely affect other threads such that they will be unable to make further progress until the first thread resumes.
Although this restriction is not entirely unreasonable
(e.g., we may assume that the number of threads roughly equals the number of physical cores, and preemption is rare),
such queues leave much to be desired.
As with explicit spin locks, lack of true lock-freedom results in suboptimal performance unless preemption is disabled. (This can be undesirable or harder to do, especially in user space.)
Also, such queues cannot be safely used outside thread contexts, e.g.,
OS interrupts. Specialized libraries~\cite{liblfds} acknowledged~\cite{ringdisapp} this problem by
rolling back to Michael \& Scott's classical FIFO lock-free queue~\cite{Michael:1998:NAP:292022.292026}, which is correct and easily portable to all platforms but slow.

Typically, true non-blocking FIFO queues are implemented using \textit{Head} and \textit{Tail} references, which are updated using the compare-and-swap (CAS) instruction. However, CAS-based approaches do not scale well as the contention grows~\cite{Morrison:2013:FCQ:2442516.2442527,nikolaev:LIPIcs:2019:11335,Yang:2016:WQF:2851141.2851168} since \textit{Head} and \textit{Tail} have to be updated inside a CAS loop that can fail and repeat. Thus, previous works explored fetch-and-add (F\&A) on the contended parts of FIFO queues: \textit{Head} and \textit{Tail} references. F\&A always succeeds and consequently scales better. Using F\&A typically implies that there exist some ring buffers underneath.
Thus, prior works have focused on making these ring buffers efficient. However, ring buffer design through F\&A is not trivial when true lock- or wait-free progress is required. In fact, lock-free ring buffers historically needed CAS~\cite{10.1145/3332466.3374508,Tsigas:2001:SFS:378580.378611}.

Until recently, efficient ring buffers~\cite{Morrison:2013:FCQ:2442516.2442527,Yang:2016:WQF:2851141.2851168} were livelock-prone, and researchers attempted to workaround livelock issues by using additional layers of slower CAS-based queues. SCQ~\cite{nikolaev:LIPIcs:2019:11335} took a different approach by constructing a truly lock-free ring buffer that uses F\&A. Unfortunately, SCQ still lacks stronger wait-free progress guarantees.

The literature presents many approaches for building wait-free data structures. Kogan \& Petrank's \textit{fast-path-slow-path} methodology~\cite{Kogan:2012:MCF:2145816.2145835} uses a lock-free procedure for the fast path, taken most of the time, and falls back to a wait-free procedure if the fast path does not succeed. However, the methodology only considers CAS, and the construction of algorithms that heavily rely on F\&A for improved performance is unclear. (Substituting F\&A with a more general CAS construct would erase performance advantages.)

To that end, Yang and Mellor-Crummey's (YMC)~\cite{Yang:2016:WQF:2851141.2851168} wait-free queue implemented its own fast-path-slow-path method. But, as pointed out by Ramalhete and Correia~\cite{pedroWFQUEUEFULL}, YMC's design is flawed in its memory reclamation approach which, strictly described, forfeits wait-freedom. This means that a user still has to choose
from other wait-free queues which are slower, such as Kogan \& Petrank's~\cite{kpWFQUEUE} and CRTurn~\cite{pedroWFQUEUE} queues. These queues do not use F\&A and scale poorly.

This prior experience reveals that solving multiple problems such as creating an unbounded queue, avoiding livelocks, and attaining wait-freedom in the same algorithm may not be an  effective design strategy. Instead, a compartmentalized design approach may be more effective wherein at each step, we solve only one problem because it enables reasoning about separate components in isolation. For example, lock-free SCQ achieves great performance and memory efficiency. SCQ can be extended to attain wait-freedom. 

In this paper, we present Wait-Free Circular Queue (or wCQ). wCQ uses its own variation of fast-path-slow-path methodology (Section~\ref{sec:wfcq}) to attain wait-freedom and bound memory usage.
When falling back to the slow path after bounded number of attempts on the fast path, wCQ's threads collaborate with each other to ensure wait-freedom.
wCQ requires double-width CAS, which is nowadays widespread (i.e., x86 and ARM/AArch64). However, we also illustrate how wCQ can be implemented on architectures that lack such instructions including PowerPC and MIPS (see Section~\ref{sec:llsc}). We analyze wCQ's properties including  wait-freedom and bounded memory usage (Section~\ref{sec:correctness}). Our evaluations on x86 and PowerPC architectures validate wCQ's excellent performance and memory efficiency (Section~\ref{sec:eval}). Additionally, they show that wCQ's performance closely matches SCQ's. 

wCQ is the first fast wait-free queue with bounded memory usage. 
Armed with wCQ, the final goal of creating an unbounded wait-free queue is also realistic, as existing slower wait-free queues (e.g., CRTurn~\cite{pedroWFQUEUE}) can link faster wait-free ring buffers together.

\section{Preliminaries}
\label{sec:background}

For the sake of presentation, we assume a \textbf{sequentially consistent} memory model~\cite{Lamport:1979:MMC:1311099.1311750}, as otherwise, the pseudo-code becomes cluttered with implementation-specific barrier code. Our implementation inserts memory barriers wherever strongly-ordered shared memory writes are necessary.

\subsubsection*{Wait-Free Progress.}
In the literature, several categories of non-blocking data structures are
considered. In \textit{obstruction-free} structures,
progress is only guaranteed if one thread runs in isolation from others.
In \textit{lock-free} structures, \emph{one} thread is always guaranteed to make progress in a
finite number of steps. Finally, in \textit{wait-free} structures,
\emph{all} threads are always guaranteed to make progress in a
finite number of steps.
In lock-free structures, individual threads may
starve, whereas wait-freedom also implies starvation-freedom.
Unsurprisingly, wait-free data structures are
the most challenging to design, but they are especially useful for latency-sensitive applications which often have quality of service constraints.

Memory usage is an important property that is sometimes overlooked in the design of non-blocking data structures. All \emph{truly} non-blocking algorithms must bound memory usage. Otherwise, no further progress can be made when memory is exhausted.
Therefore, such algorithms are inherently blocking. True wait-freedom, thus, also implies
bounded memory usage. Despite many prior attempts to design wait-free FIFO queues, the design of
\emph{high-perform\-ant} wait-free queues, which also have bounded memory usage,
remains challenging -- exactly the problem that we solve in this paper.

\subsubsection*{Read-Modify-Write.}
Lock-free and wait-free algorithms typically use read-modify-write (RMW)
operations, which atomically read a memory variable, perform some operation 
on it, and write back the result.
Modern CPUs implement RMWs via
compare-and-swap (CAS) or a pair of
load-link (LL)/store-conditional (SC) instructions. For better scalability,
some architectures, such as x86-64~\cite{intel:manual} and AArch64~\cite{arm:manual},
support specialized operations such as F\&A (fetch-and-add), OR (atomic OR),
and XCHG (exchange). 
The execution time of these specialized instructions is bounded, which makes them especially useful for wait-free algorithms. Since the above operations are
unconditional, even architectures that merely implement them
via LL/SC (but not CAS) may still achieve bounded execution time depending
on whether their LL reservations occur in a wait-free manner. (Note that LL/SC can still fail due to OS events, e.g., interrupts.)
x86-64 and AArch64 also support (double-width) CAS on two
\textit{contiguous} memory words, which we refer to as CAS2.
CAS2 must not be confused with double-CAS, which updates two
\textit{arbitrary} words but is rarely
available in commodity hardware.

\subsubsection*{Infinite Array Queue.}
Past works~\cite{Morrison:2013:FCQ:2442516.2442527,nikolaev:LIPIcs:2019:11335,Yang:2016:WQF:2851141.2851168} argued that F\&A scales \emph{significantly} better than a CAS loop under contention and proposed to use ring buffers. However, building a correct and performant ring buffer with F\&A is challenging.

\begin{figure}
\begin{subfigure}{.5\columnwidth}
{\LinesNotNumbered
\begin{algorithm2e}[H]
\textbf{void} *Array[INFINITE\_SIZE]\;
\textbf{int} Tail = 0, Head = 0\;
\Fn{\textbf{void} Enqueue(\textbf{void *} p)} {
\While {\upshape True} {
    T = F\&A(\&Tail, 1)\;
	v = XCHG(\&Array[T], p)\;
	\lIf {\upshape v = $\bot$} {
		\Return
	}
}
}
\end{algorithm2e}
}
\end{subfigure}%
\hspace{-1.5em}
\begin{subfigure}{.5\columnwidth}
{\LinesNotNumbered
\begin{algorithm2e}[H]
\setcounter{AlgoLine}{7}
\Fn{\textbf{void} *Dequeue()} {
\Do (\tcp*[f]{While not empty}){\upshape{Load(\&Tail) > H + 1}} {
    H = F\&A(\&Head, 1)\;
    p = XCHG(\&Array[H], $\top$)\;
    \lIf {\upshape p $\neq$ $\bot$} {\Return p}
}
\Return \Null\;
}
\end{algorithm2e}
}
\end{subfigure}%
\vspace{-7pt}
\caption{Livelock-prone infinite array queue.}
\label{alg:infring1}
\end{figure}

To help with reasoning about this problem, literature~\cite{10.1145/78969.78972,Morrison:2013:FCQ:2442516.2442527} describe the concept of an ``infinite array queue.'' The queue~\cite{Morrison:2013:FCQ:2442516.2442527} assumes that we have an infinite array (i.e., memory is unlimited) and old array entries need not be recycled (Figure~\ref{alg:infring1}). The queue is initially empty, and
its array entries are set to a reserved $\bot$ value. When calling \textit{Enqueue}, \textit{Tail} is incremented using F\&A. The value returned by F\&A will point to an array entry where a new element can be inserted.  \textit{Enqueue} succeeds if the previous value was $\bot$. Otherwise, some dequeuer
arrived before the corresponding enqueuer, 
and thus the entry was modified to prevent the enqueuer
from using it. In the latter case, \textit{Enqueue} will attempt to execute F\&A again.
\textit{Dequeue} also executes F\&A on \textit{Head}. F\&A retrieves the position of a previously enqueued entry. To prevent this entry from being reused by a late enqueuer, \textit{Dequeue} places a reserved $\top$ value. If the previous value is not $\bot$, \textit{Dequeue} returns that value, which was previously inserted by \textit{Enqueue}. Otherwise, \textit{Dequeue} attempts to execute F\&A again.

Note that this infinite queue is clearly susceptible to livelocks since all \textit{Dequeue} calls may get well ahead of \textit{Enqueue} calls, preventing any further progress. Prior queues, such as LCRQ~\cite{Morrison:2013:FCQ:2442516.2442527}, YMC~\cite{Yang:2016:WQF:2851141.2851168}, and SCQ~\cite{nikolaev:LIPIcs:2019:11335} are all inspired by this infinite queue, and they all implement ring buffers that use finite memory. However, they diverge on how to address the livelock problem. LCRQ workarounds livelocks by ``closing'' stalled ring buffers. YMC tries to hit two birds with one stone by assisting unlucky threads while also ensuring wait-freedom (but the approach is flawed as previously discussed). Finally, SCQ uses a special ``threshold'' value to construct a lock-free ring buffer instead. In this paper, we use and extend SCQ's approach to create a wait-free ring buffer.

\subsubsection*{SCQ Algorithm.}
In SCQ~\cite{nikolaev:LIPIcs:2019:11335}, a program comprises of $k$ threads, and the ring buffer size is $n$, which is a power of two. SCQ also reasonably assumes that $k\le n$. wCQ, discussed in this paper, makes identical assumptions.

\begin{figure}
\begin{minipage}{.5\columnwidth}
{\LinesNotNumbered
\begin{algorithm2e}[H]
\Fn{\textbf{bool} Enqueue\_Ptr(\textbf{void} *p)} {
\textbf{int} index = \textbf{fq}.Dequeue()\;
\If (\tcp*[f]{Full}) {\upshape index = $\varnothing$} {\Return False\;}
data[index] = p\;
\textbf{aq}.Enqueue(index)\;
\Return True\tcp*{Done}
}
\end{algorithm2e}
}
\end{minipage}%
\begin{minipage}{.5\columnwidth}
{\LinesNotNumbered
\begin{algorithm2e}[H]
\setcounter{AlgoLine}{7}
\Fn{\textbf{void} *Dequeue\_Ptr()} {
\textbf{int} index = \textbf{aq}.Dequeue()\;
\If (\tcp*[f]{Empty}) {\upshape index = $\varnothing$} {\Return \Null\;}
\textbf{void} *p = data[index]\;
\textbf{fq}.Enqueue(index)\;
\Return p\tcp*{Done}
}
\end{algorithm2e}
}
\end{minipage}%
\vspace{-7pt}
\caption{Storing pointers via indirection.}
\label{alg:ringptr}
\end{figure}

One of SCQ's key ideas is indirection. SCQ uses two ring buffers which store
``allocated'' and ``freed'' indices in \textbf{aq} and \textbf{fq}, respectively.
Data entries (which are pointers or any fixed-size data) are stored in a separate array
which is referred to by these indices. For example, to store pointers (Figure~\ref{alg:ringptr}), \textit{Enqueue\_Ptr} dequeues an index
from {\bf fq}, initializes an array entry with a pointer, and inserts
the index into {\bf aq}. A counterpart \textit{Dequeue\_Ptr} dequeues this index from {\bf aq},
reads the pointer, and inserts the index back into {\bf fq}.

SCQ's \textit{Enqueue} (for either \textbf{aq} or \textbf{fq}) need not check if the corresponding queue is full since the maximum number of indices is $n$. This greatly simplifies SCQ's design since only \textit{Dequeue} needs to be handled specially for an empty queue.
SCQ is \textit{operation-wise} lock-free: at least one enqueuer and one dequeuer both succeed after a finite number of steps.

Figure~\ref{alg:ringadv} presents the SCQ algorithm. The algorithm provides a practical implementation of the infinite array queue that was previously discussed. Since memory is finite in practice, the same array entry can be used over and over again in a cyclic fashion. Thus, to distinguish recycling attempts, SCQ records \textit{cycles}. XCHG is replaced with CAS as the same entry can be accessed through different cycles. LCRQ~\cite{Morrison:2013:FCQ:2442516.2442527} uses a very similar idea.

Each ring buffer in SCQ keeps its {\it Head} and {\it Tail} references. Using these references, \textit{Enqueue} and \textit{Dequeue} can
determine a position in the ring buffer ($j=Tail\mod n$) and a current cycle ($cycle=Tail\div n$).
Each \textit{entry} in SCQ ring buffers records \textit{Cycle} and \textit{Index} for the inserted entry. Since each entry has an implicit position, only the \textit{Cycle} part needs to be recorded.
(Note that \textit{Index}'s bit-length matches the position bit-length since $n$ is the maximum value for both; thus \textit{Index} and \textit{Cycle} still fit in a single word.) SCQ also reserves one bit in each entry: \textit{IsSafe} handles corner cases when an entry is still occupied by some old cycle but some newer dequeuer needs to mark the entry unusable.

All entries are updated sequentially. To reduce false sharing, SCQ permutes queue
positions by using {\it Cache\_Remap}, which places two adjacent entries
into different cache lines; the same
cache line is not reused as long as possible.

When \textit{Enqueue} successfully inserts a new index, it resets a special \textit{Threshold} value to the maximum. The threshold value is decremented by \textit{Dequeue} when the latter is unable to make progress. A combination of these design choices, justified in~\cite{nikolaev:LIPIcs:2019:11335}, allows SCQ to achieve lock-freedom directly inside the ring buffer itself.

\begin{figure*}
\begin{subfigure}{.32\textwidth}
\begin{algorithm2e}[H]
\textbf{int} Threshold = -1\tcp*{Empty SCQ $ $ }
\textbf{int} Tail = 2n, Head = 2n\;
\tcp {Init entries: \{.Cycle=0,}
\tcp {  .IsSafe=1, .Index=$\bot$\}}
\textbf{entry\_t} Entry[2n]\;
\BlankLine
\BlankLine
\BlankLine
\BlankLine
\BlankLine
\BlankLine
\BlankLine
\BlankLine
\BlankLine
\BlankLine
\Fn{\textbf{void} Enqueue\_SCQ(\textbf{int} index)} {
\While {\upshape try\_enq(index) $\ne$ OK} {
	\tcp{Try again}
}
}
\BlankLine
\Fn{\textbf{int} Dequeue\_SCQ()} {
\If {\upshape Load(\&Threshold) $<$ 0} {
	\Return $\varnothing$\tcp*{Empty}
}
\While {\upshape try\_deq(\&idx) $\ne$ OK}{
\tcp{Try again}
}
\Return idx\;
}
\end{algorithm2e}
\end{subfigure}%
\hspace{-1.5em}
\begin{subfigure}{.35\textwidth}
\begin{algorithm2e}[H]
\setcounter{AlgoLine}{10}
\Fn {\textbf{void} consume(\textbf{int} h, \textbf{int} j, \textbf{entry} e)} {
	OR(\&Entry[j], \{ 0, 0, $\bot_c$ \})\;
}
\Fn {\textbf{void} catchup(\textbf{int} tail, \textbf{int} head)} {
\While {\upshape \Not CAS(\&Tail, tail, head)} {
	head = Load(\&Head)\;
	tail = Load(\&Tail)\;
	\lIf {\upshape tail $\ge$ head} {\textbf{break}}
} 
}
\Fn{\textbf{int} try\_enq(\textbf{int} index)} {
T = F\&A(\&Tail, 1)\;
j = Cache\_Remap(T \ModOp 2n)\;
E = Load(\&Entry[j])\;\label{ealoop}
	\If {\upshape E.Cycle < Cycle(T) $\AndOp$ (E.IsSafe \OrOp Load(\&Head) $\le$ T) $\AndOp$ (E.Index = $\bot$ \OrOp $\bot_c$)} {
New = \{ Cycle(T), 1, index \}\;
\If {\upshape \Not CAS(\&Entry[j], E, New)} {\Goto{ealoop}}
	\If {\upshape Load(\&Threshold) $\ne$ 3n-1} {Store(\&Threshold, $3n - 1$)}
\Return OK\tcp*{Success$ $ }
}
\Return T\tcp*{Try again$ $ }
}
\end{algorithm2e}
\end{subfigure}%
\begin{subfigure}{.32\textwidth}
\begin{algorithm2e}[H]
\setcounter{AlgoLine}{29}
\Fn{\textbf{int} try\_deq(\textbf{int} *index)} {
H = F\&A(\&Head, 1)\;
j = Cache\_Remap(H \ModOp 2n)\;
E = Load(\&Entry[j])\;\label{daloop}
\If {\upshape E.Cycle = Cycle(H)} {
	consume(H, j, E)\;
	*index = E.Index\;
	\Return OK\tcp*{Success}
}
New = \{ E.Cycle, 0, E.Index \}\;
\If {\upshape E.Index = $\bot$ \OrOp $\bot_c$} {
	New = \{ Cycle(H), E.IsSafe, $\bot$\}\;
}
\If {\upshape E.Cycle < Cycle(H)} {
\lIf {\upshape \Not CAS(\&Entry[j], E, New)} {\Goto{daloop}}
}
T = Load(\&Tail)\tcp*{Exit if}
\If (\tcp*[f]{empty}) {\upshape T $\le$ H + 1} {
	catchup(T, H + 1)\;
	F\&A(\&Threshold, -1)\;
	*index = $\varnothing$\tcp*{Empty}
	\Return OK\tcp*{Success}
}
\If {\upshape F\&A(\&Threshold, -1) $\le$ 0} {
	*index = $\varnothing$\tcp*{Empty}
	\Return OK\tcp*{Success}
}
\Return H\tcp*{Try again}
}
\end{algorithm2e}
\end{subfigure}%
\vspace{-7pt}
\caption{Lock-free circular queue (SCQ): \textit{Enqueue} and \textit{Dequeue} are identical for both \textbf{aq} and \textbf{fq}.}
\label{alg:ringadv}
\end{figure*}

The SCQ paper~\cite{nikolaev:LIPIcs:2019:11335} provides a justification for
the maximum threshold value when using the
original infinite array queue
as well as when using a practical finite queue (SCQ). For the infinite queue,
the threshold is $2n-1$, where $n$ is the maximum number of entries.
This value is obtained by observing that the last dequeuer is at most $n$ slots
away from the last inserted entry; additionally there may be at most $n-1$ dequeuers that precede the last dequeuer.

Finite SCQ, however, is more intricate. To retain lock-freedom, it additionally
requires to
\textbf{double the capacity} of ring buffers, i.e., it
allocates  $2n$ entries while only ever using $n$ entries at any point of time.
Consequently, SCQ also needs to increase the threshold value to $3n-1$ since the last dequeuer can be $2n$ slots away from the last inserted entry. We \textbf{retain same threshold and capacity} in wCQ. 

Although not done in the original SCQ algorithm, for convenience, we distinguish two cases for \textit{Dequeue} in Figure~\ref{alg:ringadv}. When \textit{Dequeue} arrives prior to its \textit{Enqueue} counterpart, it places $\bot$. If \textit{Dequeue} arrives on time, when the entry can already be \textit{consumed}, it inserts $\bot_c$. The distinction between $\bot_c$ and $\bot$ will become useful when discussing wCQ. SCQ assigns $\bot_c=2n-1$ so that all lower bits of the number are ones ($2n$ is a power of two). This is convenient when consuming elements in Line~12, which can simply execute an atomic OR operation to replace the index with $\bot_c$ while keeping all other bits intact. We further assume that $\bot=2n-2$. Neither $\bot_c$ nor $\bot$  overlaps with any actual indices $[0..n-1]$.

\subsubsection*{Kogan \& Petrank's method.}
This wait-free method~\cite{Kogan:2012:MCF:2145816.2145835} implies
that a fast path algorithm (ideally with a performance similar to a lock-free algorithm) is attempted multiple times (MAX\_PATIENCE
in our paper). When no progress is made, a slow path algorithm will
guarantee completion after a finite number of steps, though with a higher performance cost. The slow path expects collaboration from other threads. Periodically, when performing data structure operations, \textit{every} thread checks if any other thread needs helping.

The original methodology dynamically allocates slow path helper descriptors, which need to be reclaimed. But dynamic memory allocation makes it trickier to guarantee bounded memory usage, as experienced by YMC. Also, it is not clear how to apply the methodology when using specialized (e.g., F\&A) instructions.
In this paper, we address these issues for SCQ.

\section{Wait-Free Circular Queue (wCQ)}
\label{sec:wfcq}

wCQ's key insight is to avoid memory reclamation issues altogether. Because wCQ
only needs per-thread descriptors and the ring buffer itself, it does not need
to deal with dynamic memory allocation.
The original Kogan \& Petrank's fast-path-slow-path methodology cannot be used as-is due to memory reclamation concerns as well as lack of F\&A support. Instead, wCQ uses a variation of this methodology specifically designed for SCQ. All threads collaborate to guarantee wait-free progress.

\subsubsection*{Assumptions.}

Generally speaking, wCQ requires double-width CAS (CAS2) to properly synchronize concurrent queue alterations. wCQ also assumes hardware-implemented F\&A and atomic OR to guarantee wait-freedom. However, these requirements are not very strict. (See Section~\ref{sec:hwsupport} and Section~\ref{sec:llsc} for more details.)

\subsubsection*{Data Structure.}

Figure~\ref{alg:wfcqstruct} shows changes to the SCQ structures, described in Section~\ref{sec:background}. Each entry is extended to a \textit{pair}. A pair comprises of the original entry \textit{Value} and a special \textit{Note}, discussed later. Each thread maintains per-thread state, the \textbf{thrdrec\_t} record. Its private fields are only used by the thread when it attempts to \textit{assist} other threads. In contrast, its shared fields are used to \textit{request help}.

\subsection{Bird's-Eye View}

Figure~\ref{alg:wfcq} shows the \textit{Enqueue\_wCQ} and \textit{Dequeue\_wCQ} procedures.
\textit{Enqueue\_wCQ} first checks if any other thread needs help by calling \textit{help\_threads}, after which it attempts to use the fast path to insert an entry. The fast path is identical to \textit{Enqueue\_SCQ}. When exceeding MAX\_PATIENCE iterations, \textit{Enqueue\_wCQ} takes
the slow path. In the slow path, \textit{Enqueue\_wCQ}  requests help by recording its
last \textit{Tail} value that was tried (in \textit{initTail} and \textit{localTail}) and the \textit{index} input parameter. The \textit{initTail} and \textit{localTail} are initially identical and only diverge later (see \textit{slow\_F\&A} below). Additionally,
\textit{Enqueue\_wCQ} sets extra flags to indicate an active enqueue help request. Since the entire request is to be read atomically, we use \textit{seq1} and \textit{seq2} to verify the integrity of reads. If \textit{seq1} does not match \textit{seq2}, no valid request exists for the thread. Each time a slow path is complete, \textit{seq1} is incremented (Line~28). Each time a new request is made, \textit{seq2} is set to \textit{seq1} (Line~24). Subsequently, \textit{enqueue\_slow} is called to take the slow path.

\begin{figure}
\begin{subfigure}{.55\columnwidth}
{\LinesNotNumbered
\begin{algorithm2e}[H]
\textbf{struct} phase2rec\_t \{

    \hspace{.5em}\textbf{int} seq1\tcp*{= 1}
    \hspace{.5em}\textbf{int} *local\;
    \hspace{.5em}\textbf{int} cnt\;
    \hspace{.5em}\textbf{int} seq2;\hspace{3em}\}\tcp*{= 0}
\BlankLine
\textbf{struct} entpair\_t \{

    \hspace{.5em}\textbf{int} Note\tcp*{ = -1}
    \hspace{.5em}\textbf{entry\_t} Value\tcp*{= \{ .Cycle=0,}
    \}\tcp*{.IsSafe=1, .Enq=1, .Index=$\bot$ \}}
    \BlankLine
\textbf{entpair\_t} Entry[2n]\;
\textbf{thrdrec\_t} Record[NUM\_THRDS]\;
\BlankLine
\textbf{int} Threshold = -1\tcp*{Empty wCQ}
\textbf{int} Tail = 2n, Head = 2n\;
\end{algorithm2e}}
\end{subfigure}%
\hspace{-2em}
\begin{subfigure}{.5\columnwidth}
{\LinesNotNumbered
\begin{algorithm2e}[H]
\textbf{struct} thrdrec\_t \{
    
    \hspace{.5em}\tcp{=== Private Fields ===}
    \hspace{.5em}\textbf{int} nextCheck\tcp*{= HELP\_DELAY}
    \hspace{.5em}\textbf{int} nextTid\tcp*{Thread ID}
    \hspace{.5em}\tcp{=== Shared Fields ===}
    \hspace{.5em}\textbf{phase2rec\_t} phase2\tcp*{Phase 2}
    \hspace{.5em}\textbf{int} seq1\tcp*{= 1}
    \hspace{.5em}\textbf{bool} enqueue\;
    \hspace{.5em}\textbf{bool} pending\tcp*{= false}
    \hspace{.5em}\textbf{int} localTail, initTail\;
    \hspace{.5em}\textbf{int} localHead, initHead\;
    \hspace{.5em}\textbf{int} index\;
    \hspace{.5em}\textbf{int} seq2\hspace{3em}\}\tcp*{= 0}
\end{algorithm2e}}
\end{subfigure}%
\vspace{-7pt}
\caption{Auxiliary structures.}
\label{alg:wfcqstruct}
\end{figure}

A somewhat similar procedure is used for \textit{Dequeue\_wCQ}, with the exception that \textit{Dequeue\_wCQ} also needs to check if the queue is empty. After completing the slow path, the output result needs to be gathered. In SCQ, output is merely consumed by using atomic OR. In wCQ, \textit{consume} is extended to mark all pending enqueuers, as discussed below.

\subsubsection*{Helping Procedures.}
Figure~\ref{alg:wfcqhelp} shows the \textit{help\_threads} code, which is inspired by the method of 
Kogan \& Petrank~\cite{Kogan:2012:MCF:2145816.2145835}. Each thread skips HELP\_DELAY iterations using \textit{nextCheck} to amortize the cost of \textit{help\_threads}. The procedure circularly iterates across all threads, \textit{nextTid}, to find ones with a pending helping request. Finally, \textit{help\_threads} calls \textit{help\_enqueue} or \textit{help\_dequeue}. A help request is atomically loaded and passed to \textit{enqueue\_slow} and \textit{dequeue\_slow}.

\begin{figure}
\begin{subfigure}{.53\columnwidth}
\begin{algorithm2e}[H]
\Fn {\textbf{void} consume(\textbf{int} h, \textbf{int} j, \textbf{entry\_t} e)} {
        \lIf {\Not e.Enq} {finalize\_request(h)}
        OR(\&Entry[j].Value, \{ 0, 0, 1, $\bot_c$ \})\;
}
\Fn {\textbf{void} finalize\_request(\textbf{int} h)} {
                i = (TID + 1) \ModOp NUM\_THRDS\;
                \While {\upshape i != TID} {
                        \textbf{int} *tail = \&Record[i].localTail\;
                        \If {\upshape Counter(*tail) = h} {
                                CAS(tail, h, h | FIN)\;
                                \Return\;
                        }
                        i = (i + 1) \ModOp NUM\_THRDS\;
                }
}
\Fn{\textbf{void} Enqueue\_wCQ(\textbf{int} index)} {
help\_threads()\;
\tcp{== Fast path (SCQ) ==}
\textbf{int} count = MAX\_PATIENCE\;
\While {\upshape $--$count $\ne$ 0}{
tail = try\_enq(index)\;
\lIf {\upshape tail = OK} {\Return}
}
\tcp{== Slow path (wCQ) ==}
\textbf{thrdrec\_t} *r = \&Record[TID]\;
\textbf{int} seq = r->seq1\;
r->localTail = tail\;
r->initTail = tail\;
r->index = index\;
r->enqueue = true\;
r->seq2 = seq\;
r->pending = true\;
enqueue\_slow(tail, index, r)\;
r->pending = false\;
r->seq1 = seq + 1\;
}
\end{algorithm2e}
\end{subfigure}%
\hspace{-1.5em}
\begin{subfigure}{.5\columnwidth}
\begin{algorithm2e}[H]
\setcounter{AlgoLine}{28}
\Fn{\textbf{int} Dequeue\_wCQ()} {
\If {\upshape Load(\&Threshold) $<$ 0} {
        \Return $\varnothing$\tcp*{Empty}
}
help\_threads()\;
\tcp{== Fast path (SCQ) ==}
\textbf{int} count = MAX\_PATIENCE\;
\While {\upshape $--$count $\ne$ 0}{
\textbf{int} idx\;
head = try\_deq(\&idx)\;
\lIf {\upshape head = OK} {\Return idx}
}
\tcp{== Slow path (wCQ) ==}
\textbf{thrdrec\_t} *r = \&Record[TID]\;
\textbf{int} seq = r->seq1\;
r->localHead = head\;
r->initHead = head\;
r->enqueue = false\;
r->seq2 = seq\;
r->pending = true\;
dequeue\_slow(head, r)\;
r->pending = false\;
r->seq1 = seq + 1\;
\tcp{Get slow-path results}
h = Counter(r->localHead)\;
j = Cache\_Remap(h \ModOp 2n)\;
Ent = Load(\&Entry[j].Value)\;
\If {\upshape Ent.Cycle = Cycle(h) \AndOp Ent.Index $\ne$ $\bot$} {
        consume(h, j, Ent)\;
        \Return Ent.Index\tcp*{Done}
}
\Return $\varnothing$
}
\end{algorithm2e}
\end{subfigure}%
\vspace{-7pt}
\caption{Wait-free circular queue (wCQ).}
\label{alg:wfcq}
\end{figure}

\subsection{Slow Path}

\label{sec:slowpathwcq}

Either a helpee or its helpers eventually call \textit{enqueue\_slow} and \textit{dequeue\_slow}. wCQ's key idea is that eventually all active threads assist a thread that is stuck if progress is not made. One of these threads will eventually succeed due to the underlying SCQ's lock-free guarantees. However, all helpers should repeat \textit{exactly} the same procedure as the helpee. This can be challenging since the ring buffer keeps changing.

More specifically, multiple \textit{enqueue\_slow} calls are to avoid inserting the same element multiple times into different positions. Likewise, \textit{dequeue\_slow} should only consume one element.
Figure~\ref{alg:wfcqslow} shows wCQ's approach for this problem.

In Figure~\ref{alg:wfcqslow}, a special \textit{slow\_F\&A} operation
substitutes F\&A from the fast path. The key idea is that for any
given helpee and its helpers, the global \textit{Head} and \textit{Tail}
values need to be changed only once per each iteration across all cooperative
threads (i.e., a helpee and its helpers). To support this, each
thread record maintains \textit{initTail}, \textit{localTail}, \textit{initHead}, and \textit{localHead} values. These values are initialized from the
last tail and head values from the fast path accordingly. In the beginning,
the init and local values are identical. The init value is a starting point for \textit{all} helpers (Lines~17~and~23, Figure~\ref{alg:wfcqhelp}). The local value represents the last value in \textit{slow\_F\&A}.
To support \textit{slow\_F\&A}, we redefine the
global \textit{Head} and \textit{Tail} values
to be \emph{pairs} of counters with pointers rather than just counters. (The pointer component is initially \Null.) Fast-path
procedures only use hardware F\&A on counters leaving pointers intact. However, slow-path procedures use the pointer component to store the second phase request (see below). We use the fact that paired counters are monotonically increasing to prevent the ABA problem for pointers.

\begin{figure}
\begin{subfigure}{.5\columnwidth}
\begin{algorithm2e}[H]
\Fn{\textbf{void} help\_threads()} {
\textbf{thrdrec\_t} *r = \&Record[TID]\;
\If {\upshape $--$r->nextCheck $\ne$ 0} {
        \Return\;
}
thr = \&Record[r->nextTid]\;
\If {\upshape thr->pending} {
        \If {\upshape thr->enqueue} {
                help\_enqueue(thr)\;
        }
        \Else {
                help\_dequeue(thr);\
        }
}
r->nextCheck = HELP\_DELAY\;
r->nextTid = (r->nextTid + 1) \ModOp NUM\_THRDS\;
}
\end{algorithm2e}
\end{subfigure}%
\hspace{-1.5em}
\begin{subfigure}{.54\columnwidth}
\begin{algorithm2e}[H]
\setcounter{AlgoLine}{12}
\Fn{\textbf{void} help\_enqueue(\textbf{thrdrec\_t} *thr)} {
\textbf{int} seq = thr->seq2\;
\textbf{bool} enqueue = thr->enqueue\;
\textbf{int} idx = thr->index\;
\textbf{int} tail = thr->initTail\;
\If {\upshape enqueue \AndOp thr->seq1 = seq} {
        enqueue\_slow(tail, idx, thr)\;
}
}
\Fn{\textbf{void} help\_dequeue(\textbf{thrdrec\_t} *thr)} {
\textbf{int} seq = thr->seq2\;
\textbf{bool} enqueue = thr->enqueue\;
\textbf{int} head = thr->initHead\;
\If {\upshape \Not enqueue \AndOp thr->seq1 = seq} {
        dequeue\_slow(head, thr)\;
}
}
\end{algorithm2e}
\end{subfigure}%
\vspace{-7pt}
\caption{wCQ's helping procedures.}
\label{alg:wfcqhelp}
\end{figure}

The local value is updated to the global counter (Line~25, Figure~\ref{alg:wfcqslow}) by one of the cooperative threads.
This value is compared
against the prior value stored on stack (\textit{T} or \textit{H}) such
that one and only one thread updates the local value.
Since the global value also needs to be consistently changed, we use a special protocol.
In the first phase, the local value is set to the global value with the
INC flag (Line~25). Then, the global value is incremented (Line~32).
In the second phase, the local value resets the INC flag in Line~34 (unless
it was changed by a concurrent thread in Line~86 already).

Several corner cases for \textit{try\_enq\_slow} need to be considered. One case is when an element is already inserted by another thread (Line~19, Figure~\ref{alg:wfcqslow}). Another case is when the condition in Line~6 is true. If one helper skips the entry, we want other helpers to do the same since the condition can become otherwise false at some later point. To that end, \textit{Note} is advanced to the current tail cycle, which allows Line~5 to skip the entry for later helpers.

Finally, we want to terminate future helpers if the entry is already consumed and reused for a later cycle. For this purpose, entries are inserted using the two-step procedure. We reserve an additional \textit{Enq} bit in entries. First, the entry is inserted with Enq=0. Then the helping request is finalized by setting FIN in Line~14. Any concurrent dequeuer that observes Enq=0 will help setting FIN (Line~2, Figure~\ref{alg:wfcq}). FIN is set directly to the thread record's \textit{localTail} value by stealing one bit to prevent any future increments with \textit{slow\_F\&A}. Finally, either of the two threads will flip Enq to 1, at which point the entry can be consumed.

Similar corner cases exist in \textit{try\_deq\_slow}, where \textit{Note} prevents reusing a previously invalidated entry. \textit{try\_deq\_slow} also uses FIN to terminate all dequeuers when the result is detected.
\textit{slow\_F\&A} takes care of synchronized per-request increments (by using another bit, called INC) and also terminates helpers when detecting FIN.

Since thresholds are to be decremented only once, \textit{slow\_F\&A} decrements the threshold when calling it from \textit{try\_deq\_slow}. (Decrementing \textit{prior} to dequeueing is still acceptable.)

\begin{figure*}
\begin{subfigure}{.5\textwidth}
\begin{algorithm2e}[H]
\Fn{\textbf{bool} try\_enq\_slow(\textbf{int} T, \textbf{int} index, \textbf{thrdrec\_t} *r)} {
j = Cache\_Remap(T \ModOp 2n)\;
Pair = Load(\&Entry[j])\;\label{slweloop}
Ent = Pair.Value\;
\If {\upshape Ent.Cycle < Cycle(T) $\AndOp$ Pair.Note < Cycle(T)} {
\If {\upshape \Not (Ent.IsSafe \OrOp Load(\&Head) $\le$ T) \OrOp $ $ $ $ $ $ $ $ (Ent.Index $\ne$ $\bot$, $\bot_c$)} {
N.Value = Ent\tcp*{Avert helper enqueuers from using it}
N.Note = Cycle(T)\;
\lIf {\upshape \Not CAS2(\&Entry[j], Pair, N)} {\Goto{slweloop}}
\Return False\tcp*{Try again}
}
\tcp{Produce an entry: .Enq=0}
N.Value = \{ Cycle(T), 1, 0, index \}\;
N.Note = Pair.Note\;
\lIf (\tcp*[f]{The entry has changed}) {\upshape \Not CAS2(\&Entry[j], Pair, N)} {\Goto{slweloop}}
\If (\tcp*[f]{Finalize the help request}) {\upshape CAS(\&r->localTail, T, T | FIN)} {
Pair = N\;
N.Value.Enq = 1\;
CAS2(\&Entry[j], Pair, N)\;
}
\lIf {\upshape Load(\&Threshold) $\ne$ 3n-1} {Store(\&Threshold, 3n-1)}
}
\lElseIf (\tcp*[f]{Not yet inserted by another thread}) {\upshape\textbf{(}\xspace Ent.Cycle $\ne$ Cycle(T)} {\Return False}
\Return True\tcp*{Success}
}
\Fn{\textbf{bool} slow\_F\&A(\textbf{intpair} *globalp, \textbf{int} *local,
	\textbf{int} *v, \textbf{int} *thld)} {
\textbf{phase2rec\_t} *phase2 = \&Record[TID].phase2\;
\tcp{Global Tail/Head are replaced with \{.cnt, .ptr\} pairs:}
\tcp{use only .cnt for fast paths; .ptr stores `phase2'.}
\tcp{INVARIANT: *local, *v < *global (`global' has the next value)}
\tcp{VARIABLES: `global', thread-record (`local'), on-stack (`v'):}
\tcp{`local' syncs helpers so that `global' increments only once}
\tcp{Increment `local' (with INC, Phase 1), then `global',}
\Do (\tcp*[f]{finally reset INC (Phase 2). `thld' is only used for Head.}){\upshape{\Not CAS2(globalp, \{ cnt, \Null \hspace{1pt}\}, \{ cnt + 1, phase2 \})}} {
\textbf{int} cnt = load\_global\_help\_phase2(globalp, local)\;
\If (\tcp*[f]{Phase 1}) {\upshape cnt = $\varnothing$ \OrOp \Not CAS(local, *v, cnt | INC)} {
*v = *local\;\label{faafin}
\DontPrintSemicolon
\lIf {\upshape *v \& FIN} {\Return False;\tcp*[f]{Loop exits if FIN=1}}
\lIf {\upshape \Not (*v \& INC)} {\Return True;\tcp*[f]{Already incremented}}
\PrintSemicolon
cnt = Counter(*v)\;
}
\DontPrintSemicolon
\lElse { *v = cnt | INC;\tcp*[f]{Phase 1 completes}}
\PrintSemicolon
prepare\_phase2(phase2, local, cnt)\tcp*{Prepare help request}
}
\tcp{Loops are finite, all threads eventually help the thread that is stuck}
\lIf {\upshape thld} {F\&A(thld, -1)}
  CAS(local, cnt | INC, cnt)\;
  CAS2(globalp, \{ cnt + 1, phase2 \}, \{ cnt + 1, \Null \hspace{1pt}\})\;
*v = cnt\tcp*{Phase 2 completes}
\Return True\tcp*{Success}
}
\BlankLine
\Fn{\textbf{void} prepare\_phase2(\textbf{phase2rec\_t} *phase2, \textbf{int} *local, \textbf{int} cnt)} {
\textbf{int} seq = ++phase2->seq1\;
phase2->local = local\;
phase2->cnt = cnt\;
phase2->seq2 = seq\;
}
\end{algorithm2e}
\end{subfigure}%
\begin{subfigure}{.5\textwidth}
\begin{algorithm2e}[H]
\setcounter{AlgoLine}{42}
\Fn{\textbf{int} try\_deq\_slow(\textbf{int} H, \textbf{thrdrec\_t} *r)} {
j = Cache\_Remap(H \ModOp 2n)\;
Pair = Load(\&Entry[j])\;\label{swdloop}
Ent = Pair.Value\;
\tcp{Ready or consumed by helper ($\bot_c$ or value)}
\If {\upshape Ent.Cycle = Cycle(H) \AndOp Ent.Index $\ne$ $\bot$} {
		CAS(\&r->localHead, H, H | FIN)\tcp*{Terminate helpers}
	\Return True\tcp*{Success}
}
N.Note = Pair.Note\;
Val = \{ Cycle(H), Ent.IsSafe, 1, $\bot$ \}\;
\If {\upshape Ent.Index $\ne$ $\bot$, $\bot_c$} {
\If (\tcp*[f]{Avert helper dequeuers}) {\upshape Ent.Cycle < Cycle(H) \AndOp $ $ Pair.Note < Cycle(H)} {
N.Value = Pair.Value\tcp*{from using it}
N.Note = Cycle(H)\;
r = CAS2(\&Entry[j], Pair, N)

\lIf {\upshape \Not r} {\Goto{swdloop}}
}
Val = \{ Ent.Cycle, 0, Ent.Enq, Ent.Index \}\;
}
N.Value = Val\;
\If {\upshape Ent.Cycle < Cycle(H)} {
\If {\upshape \Not CAS2(\&Entry[j], Pair, N)} {\Goto{swdloop}\;}
}
T = Load(\&Tail)\tcp*{Exit if queue}
\If (\tcp*[f]{is empty}){\upshape T $\le$ H + 1} {
	catchup(T, H + 1)\;
}
\If {\upshape Load(\&Threshold) $<$ 0} {
	CAS(\&r->localHead, H, H | FIN)\;
	\Return True\tcp*{Success}
}
\Return False\tcp*{Try again}
}
\BlankLine
\Fn{\textbf{void} enqueue\_slow(\textbf{int} T, \textbf{int} index, \textbf{thrdrec\_t} *r)} {
\While {\upshape slow\_F\&A(\&Tail, \&r->localTail, \&T, \Null)} {
\lIf {\upshape try\_enq\_slow(T, index, r)} {\Break}
}
}
\BlankLine
\Fn{\textbf{void} dequeue\_slow(\textbf{int} H, \textbf{thrdrec\_t} *r)} {
\textbf{int} *thld = \&Threshold\;
\While {\upshape slow\_F\&A(\&Head,\&r->localHead,\&H, thld)} {
\lIf {\upshape try\_deq\_slow(H, r)} {\Break}
}
}
\BlankLine
\Fn{\textbf{int} load\_global\_help\_phase2(\textbf{intpair} *globalp, \textbf{int} *mylocal)} {
		\Do (\tcp*[f]{Load globalp \& help complete Phase 2 (i.e., make .ptr=null)}){\upshape{\Not CAS2(globalp, \{ gp.cnt, phase2 \}, \{ gp.cnt, \Null \hspace{1pt}\})}} {
	\tcp{This loop is finite, see the outer loop and FIN}
\lIf (\tcp*[f]{The outer loop exits}) {\upshape *mylocal \& FIN}{\Return $\varnothing$}
  \textbf{intpair} gp = Load(globalp)\tcp*{the thread that is stuck}
  \textbf{phase2rec\_t} *phase2 = gp.ptr\;
  \lIf (\tcp*[f]{No help request, exit}) {\upshape phase2 = \Null} {\Break}
  \textbf{int} seq = phase2->seq2\;
  \textbf{int} *local = phase2->local\;
  \textbf{int} cnt = phase2->cnt\;
  \tcp{Try to help, fails if `local' was already advanced}
  \lIf {\upshape phase2->seq1 = seq} {CAS(local, cnt | INC, cnt)}
  \tcp{No ABA problem (on .ptr) as .cnt increments monotonically}
}
\Return gp.cnt\tcp*{Return the .cnt component only}
}
\end{algorithm2e}
\end{subfigure}%
\vspace{-10pt}
\caption{wCQ's slow-path procedures.}
\label{alg:wfcqslow}
\end{figure*}

\begin{figure*}[ht]
\includegraphics[width=\textwidth]{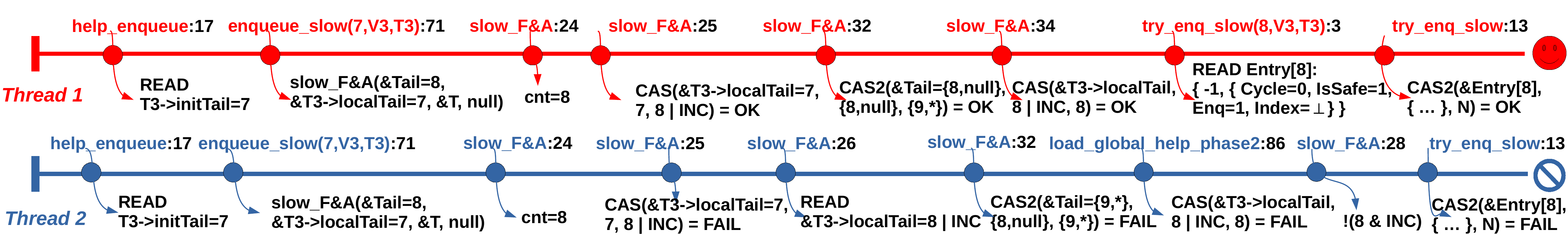}
		\caption{Enqueue slow-path execution example (ignoring CacheRemap's permutations).}
\label{fig:exec}
\end{figure*}

\subsubsection*{Second Phase Request}
When incrementing the global counter in slow\_F\&A, a thread must update
the local value to the previous global value. The thread tentatively sets
the counter in Phase 1, but the INC flag must be cleared in Phase 2, at which
point all cooperative threads will use the same local value. We request help
from cooperative threads by using the \textit{phase2} pointer. This pointer is
set simultaneously when the global value increments (Line~32, Figure~\ref{alg:wfcqslow}). The previous
value will be recorded in the \textit{phase2} request.

One corner case arises when the fast-path procedure unconditionally increments
the global value while \textit{phase2} is already set. This may cause Line~87 to
sporadically fail even though Line~86 succeeds. In this case, Line~86 is
repeated without any side effect until Line~87 succeeds or the \textit{phase2} request is finalized by some thread. All fast-path threads will
eventually be converging to help the thread that is stuck in the slow path.

\subsubsection*{Decrementing \textit{Threshold}.} One important difference with \textit{try\_deq} is that \textit{try\_deq\_slow} decrements \textit{Threshold} for every iteration inside
\textit{slow\_F\&A}. The
reason for that is that \textit{try\_deq\_slow} can also be executed concurrently by helpers. Thus, to retain the original threshold bound ($3n-1$), we must make
sure that only \textit{one} cooperative thread decrements the threshold value.
The global \textit{Head} value is an ideal source for such synchronization
since it only changes once per \textit{try\_deq\_slow} iteration across all
cooperative threads. Thus, we decrement \textit{Threshold}
\textit{prior} to the actual dequeue attempt. Note that \textit{try\_deq}
is doing it after a \textit{failure}, which is merely a performance optimization
since the $3n-1$ bound is computed in the original SCQ algorithm regardless
of the status of an operation. This performance optimization is obviously irrelevant for
the slow path.

\subsubsection*{Bounding \textit{catchup}.}
Since catchup merely improves performance in SCQ by reducing contention, there is no problem in explicitly limiting the maximum number of iterations. We do so to avoid unbounded loops in catchup.

\subsubsection*{Example.}
Figure~\ref{fig:exec} shows an enqueue scenario with three threads. Thread 3 attempts to insert index V3 but is stuck and requested help. Thread 3 does not make any progress on its own. (The corresponding thread record is denoted as T3.) Thread~1 and Thread 2 eventually observe that Thread 3 needs helping. This example shows that slow\_F\&A for both Thread~1 and Thread 2 converges to the same value (8). The global Tail value is only incremented once (by Thread 1). The corresponding entry is also only updated once (by Thread 1). For simplicity, we omit Cache\_Remap's permutations.

\subsection{Hardware Support}
\label{sec:hwsupport}
Our algorithm implies hardware F\&A and atomic OR for wait-freedom. However, they are not strictly necessary since failing F\&A in the fast path can simply force us to fall back to the slow path. There, threads eventually collaborate, bounding the execution time of F\&A, which is still used for \textit{Threshold}. Likewise, output can be consumed inside \textit{dequeue\_slow}, making it possible to emulate atomic OR with CAS while bounding execution time. We omit the discussion of these changes due to their limited practical value, as both wait-free F\&A and OR are available on x86-64 and AArch64.

\section{Architectures with ordinary LL/SC}
\label{sec:llsc}

Outside of the x86(-64) and ARM(64) realm, CAS2
may not be available. In this section, we present an
approach that we have used to implement wCQ on
PowerPC~\cite{ppc:manual} and MIPS~\cite{mips:manual}. Other architectures may also adopt this approach depending upon what operations are permitted between their LL and SC instructions.

We first note that CAS2 for global \textit{Head} and \textit{Tail} pairs can
simply be replaced with regular CAS by packing a small thread index (in lieu of the \textit{phase2} pointer) with a reduced (e.g., 48-bit rather than 64-bit) counter. However, this approach is less suitable for CAS2 used for entries
(using 32-bit counters is risky due to a high chance of an overflow).
Thus, we need an approach to store two words.

Typical architectures implement only a weak version of LL/SC, which has
certain restrictions, e.g., memory alterations between LL
and SC are disallowed. Memory reads in between, however, are still allowed for architectures such as PowerPC and MIPS. Furthermore, LL's reservation granularity is larger than just a memory word (e.g., L1 cache line for PowerPC~\cite{Sarkar:2012:SCP:2254064.2254102}).
This typically creates a problem known as
``false sharing,'' when concurrent LL/SC on unrelated variables residing in
the same cache line cause SC to succeed only for one variable.
This often requires careful consideration by the programmer to properly align
data to avoid false sharing.

In wCQ's slow-path procedures, both \textit{Value} and
\textit{Note} components
of entries need to be atomically loaded, but we only update one or
the other variable
at a time. We place two variables in the same
reservation granule by aligning the first variable on the double-word
boundary, so that only one LL/SC
pair succeeds at a time. We use a regular Load operation between LL and SC to
load the other variable atomically. We also construct an implicit memory fence
for Load by introducing artificial data dependency, which prevents reordering
of LL and Load. For the SC to succeed, the other variable from the same
reservation granule must remain intact.

In Figure~\ref{alg:llsc}, we present two replacements of CAS2.
We use corresponding versions that modify \textit{Value} or
\textit{Note} components of entries. These replacements implement weak CAS semantics (i.e., sporadic failures are possible). Also, only single-word load atomicity is guaranteed when CAS fails. Both restrictions are acceptable for wCQ.

\begin{figure}
\begin{subfigure}{.52\columnwidth}
\begin{algorithm2e}[H]
\Fn{\textbf{bool} CAS2\_Value(\textbf{entpair\_t} *Var,
\textbf{entpair\_t} Expect, \textbf{entpair\_t} New)} {
Prev.Value = LL(\&Var->Value)\;
Prev.Note = Load(\&Var->Note)\;
\lIf{\upshape Prev $\neq$ Expect} {\Return False}
\Return SC(\&Var->Value, New.Value)\;
}
\end{algorithm2e}
\end{subfigure}%
\hspace{-1.7em}
\begin{subfigure}{.52\columnwidth}
\begin{algorithm2e}[H]
\setcounter{AlgoLine}{5}
\Fn{\textbf{bool} CAS2\_Note(\textbf{entpair\_t} *Var,
\textbf{entpair\_t} Expect, \textbf{entpair\_t} New)} {
Prev.Note = LL(\&Var->Note)\;
Prev.Value = Load(\&Var->Value)\;
\lIf{\upshape Prev $\neq$ Expect} {\Return False}
\Return SC(\&Var->Note, New.Note)\;
}
\end{algorithm2e}
\end{subfigure}%
\caption{CAS2 implementation for wCQ using LL/SC.}
\label{alg:llsc}
\end{figure}

\section{Correctness}
\label{sec:correctness}

\begin{lemma}
wCQ's fast paths in \textit{Enqueue\_wCQ} and \textit{Dequeue\_wCQ} have a finite number of iterations for any thread.
\label{lemma:fast}
\end{lemma}

\begin{proof}
The number of iterations with loops containing
\textit{try\_enq} or \textit{try\_deq} are limited by MAX\_PATIENCE.
Each \textit{try\_enq} has only one loop in Line~25, Figure~\ref{alg:ringadv}.
That loop will continue as long as \textit{Entry[j]} is changed elsewhere.
If several enqueuers contend for the same entry (wrapping around), we
at most have $k\le n$ concurrent enqueuers, each of which may cause a retry.
Since each retry happens when another enqueuer succeeds and modifies the
cycle, \textit{Entry[j].Cycle} will eventually change after at most $k$
unsuccessful attempts such that the loop is terminated
(i.e., $E.Cycle < Cycle(T)$ is no longer true).
Likewise, with respect to contending dequeuers, \textit{Entry[j].Cycle} can
change for at most $k+(3n-1)\le 4n-1$ iterations (due to the threshold), at
which point all contending dequeuers will change \textit{Entry[j]} such that the loop in \textit{try\_enq} terminates (either $E.Cycle \ge Cycle(T)$ or $Load(\&Head) > T$ for non-safe entries). Additionally,
while dequeuing, for each given cycle, \textit{Index} can change \textit{once} to $\bot$ or $\bot_c$, and \textit{IsSafe} can be reset (once).
The argument is analogous for \textit{try\_deq}, which also has a similar loop
in Line~42. \textit{try\_deq} also calls \textit{catchup}, for which
we explicitly limit the number of iterations in wCQ (it is merely an
optimization in SCQ).
\end{proof}

\begin{lemma}
Taken in isolation from slow paths, wCQ's fast paths are linearizable.
\label{lemma:linfast}
\end{lemma}

\begin{proof}
Fast paths are nearly identical to the SCQ algorithm, which is linearizable.
\textit{catchup} is merely an optimization: it limits the number of
iterations and does not affect correctness arguments. \textit{consume} in wCQ, called from \textit{try\_deq}, internally calls \textit{finalize\_request} when Enq=0 (Line~2, Figure~\ref{alg:wfcq}). The purpose of that function is to
merely set the FIN flag for the corresponding enqueuer's local tail pointer (thread state) in the slow path, i.e., does not affect any global state. Enq=0 is only possible when involving slow paths.
\end{proof}

\begin{lemma}
If one thread is stuck, all other threads eventually converge to help it.
\label{lemma:stuck}
\end{lemma}

\begin{proof}
If the progress is not being made by one thread, all active threads will eventually attempt to help it (finishing their prior slow-path operations first). Some thread always succeeds since the underlying SCQ algorithm is (operation-wise) lock-free. Thus, the number of active threads in the slow path will reduce one by one until the very last stuck thread remains. All active helpers will then attempt to help that last thread, and one of them will succeed.
\end{proof}

\begin{lemma}
\textit{slow\_F\&A} does not alter local or global head/tail values as soon as a helpee or its helpers produce an entry.
\label{lemma:coop}
\end{lemma}

\begin{proof}
A fundamental property of \textit{slow\_F\&A} is that it terminates the slow
path in any remaining thread that attempts to insert the same entry as
soon as the entry is produced (Lines~71~and~75, Figure~\ref{alg:wfcqslow}). As soon as the entry is
produced (Line~14), FIN is set. FIN is checked in Line~27 prior
to changing either local or global head/tail further.
In-flight threads, can still attempt to access previous ring buffer
locations but this is harmless since previous locations are invalidated
through the \textit{Note} field.
\end{proof}

\begin{lemma}
\textit{slow\_F\&A} is wait-free bounded.
\label{lemma:slowwait}
\end{lemma}

\begin{proof}
CAS loops in \textit{slow\_F\&A} and \textit{load\_global\_help\_phase2} (Lines~23-32 and Lines~78-87, Figure~\ref{alg:wfcqslow}) are bounded since all threads will eventually converge to the slow path of one thread if no progress is being made according to Lemma~\ref{lemma:stuck}. At that point, some helper will succeed and set FIN. After that, all threads stuck in \textit{slow\_F\&A} will exit.
\end{proof}

\begin{lemma}
\textit{slow\_F\&A} decrements threshold only once per every global
head change.
\label{lemma:faa}
\end{lemma}

\begin{proof}
In Figure~\ref{alg:wfcqslow}, the threshold value is changed in Line~33.
This is only possible after the global value was incremented by one (Line~32).
\end{proof}

\begin{lemma}
Taken in isolation from fast paths, wCQ's slow paths are linearizable.
\label{lemma:linslow}
\end{lemma}

\begin{proof}
There are several critical differences in the slow path when compared
to the fast path. Specifically, \textit{try\_deq\_slow} decrements
the threshold value only once per an iteration across all cooperative threads (i.e.,
a helpee and its helpers). That follows from Lemma~\ref{lemma:faa} since global \textit{Head}
changes once per such iteration -- the whole point of \textit{slow\_F\&A}.

Cooperative threads immediately stop updating global \textit{Head} and \textit{Tail} pointers when the result is produced (Lines~71~and~75, Figure~\ref{alg:wfcqslow}), as it follows from Lemma~\ref{lemma:coop}.

Any cooperative thread always retrieves a value (the \textit{v} variable)
from \textit{slow\_F\&A}, which is either completely new, or was
previously observed by some other cooperative thread. As discussed in Section~\ref{sec:slowpathwcq}, a corner case arises when one cooperative thread
already skipped a slot with a given cycle but a different cooperative thread (which is late) attempts to use that skipped slot with the same cycle. That happens
due to the \textit{IsSafe} bit or other conditions in Line~6, Figure~\ref{alg:wfcqslow}. To guard against this scenario, the slow-path procedure maintains \textit{Note}, which guarantees that late cooperative threads will skip the slot
that was already skipped by one thread. (See Section~\ref{sec:slowpathwcq} for
more details.)
\end{proof}

\begin{theorem}
wCQ's memory usage is bounded.
\end{theorem}

\begin{proof}
wCQ is a statically allocated circular queue, and
it never allocates any extra memory. Thread record descriptors are bounded by the total number of threads.
\end{proof}

\begin{theorem}
wCQ is linearizable.
\end{theorem}

\begin{proof}
Linearizability of the fast and slow paths in separation follows from Lemmas~\ref{lemma:linfast}~and~\ref{lemma:linslow}.

The fast path is fully compatible with the slow path. There are just some
minor differences described below. The one-step procedure (Enq=1 right away) is used
in lieu of the two-step procedure (Enq=0 and then Enq=1) when producing new
entries. Likewise, the \textit{Note} field is not altered on the fast
path since it is only needed to synchronize cooperative threads, and
we only have one such thread on the fast path.
However, differences in Enq must be properly supported.
To the end, the fast path dequeuers fully support the semantics expected
by slow-path enqueuers with respect to the Enq bit. More specifically,
\textit{consume} always calls \textit{finalize\_request} for Enq=0 when
consuming an element, which helps to set the corresponding FIN
bit (Line~2, Figure~\ref{alg:wfcq}) such that it does not need to
wait until the slow-path enqueuer completes the Enq=0 to Enq=1 transition (Lines~14-17, Figure~\ref{alg:wfcqslow}).
\end{proof}

\begin{theorem}
\textit{Enqueue\_wCQ}/\textit{Dequeue\_wCQ} are wait-free.
\label{theorem:wcq}
\end{theorem}

\begin{proof}
The number of iterations on the fast path (Figure~\ref{alg:wfcq}) is finite
according to Lemma~\ref{lemma:fast}.
\textit{enqueue\_slow} (Line~26, Figure~\ref{alg:wfcq}) terminates when \textit{slow\_F\&A} (which is wait-free according to Lemma~\ref{lemma:slowwait}) returns false, which only happens when FIN is set for the tail pointer (Line~27, Figure~\ref{alg:wfcqslow}). FIN is always set (Line~14, Figure~\ref{alg:wfcqslow}) when some thread from the same helping request succeeds.

We also need to consider a case when a thread always gets unlucky and
is unable to make progress inside \textit{try\_enq\_slow} since entries
keep changing. According to Lemma~\ref{lemma:stuck}, in this situation,
all other threads will eventually converge to help the thread that is stuck.
Thus, the ``unlucky'' thread will either succeed by itself or the result will
be produced by some helper thread, at which point the \textit{slow\_F\&A} loop
is terminated.
(\textit{Dequeue\_wCQ} is wait-free by analogy.)
\end{proof}

\section{Evaluation}
\label{sec:eval}

Since in Section~\ref{sec:correctness} we already determined upper bounds for the
number of iterations and they are reasonable,
the primary goal of our evaluation is to make sure that this wCQ's stronger
wait-free progress property does not come at a significant cost, as it
is the case in present (true) wait-free algorithms such as CRTurn~\cite{pedroWFQUEUEFULL,pedroWFQUEUE}. Thus, our practical evaluation focuses
on key performance attributes that can be quantitatively and fairly
evaluated in \textit{all} existing queues -- memory efficiency and throughput.
Specifically, our goal is to demonstrate that wCQ's performance
in that respect is on par with that of SCQ since wCQ is based on SCQ.

We used the benchmark of~\cite{Yang:2016:WQF:2851141.2851168}
and extended it with the SCQ~\cite{nikolaev:LIPIcs:2019:11335} and CRTurn~\cite{pedroWFQUEUEFULL, pedroWFQUEUE} implementations. We implemented wCQ and integrated it into the test framework.

\begin{figure*}[ht]
\begin{subfigure}{.33\textwidth}
\includegraphics[width=\textwidth]{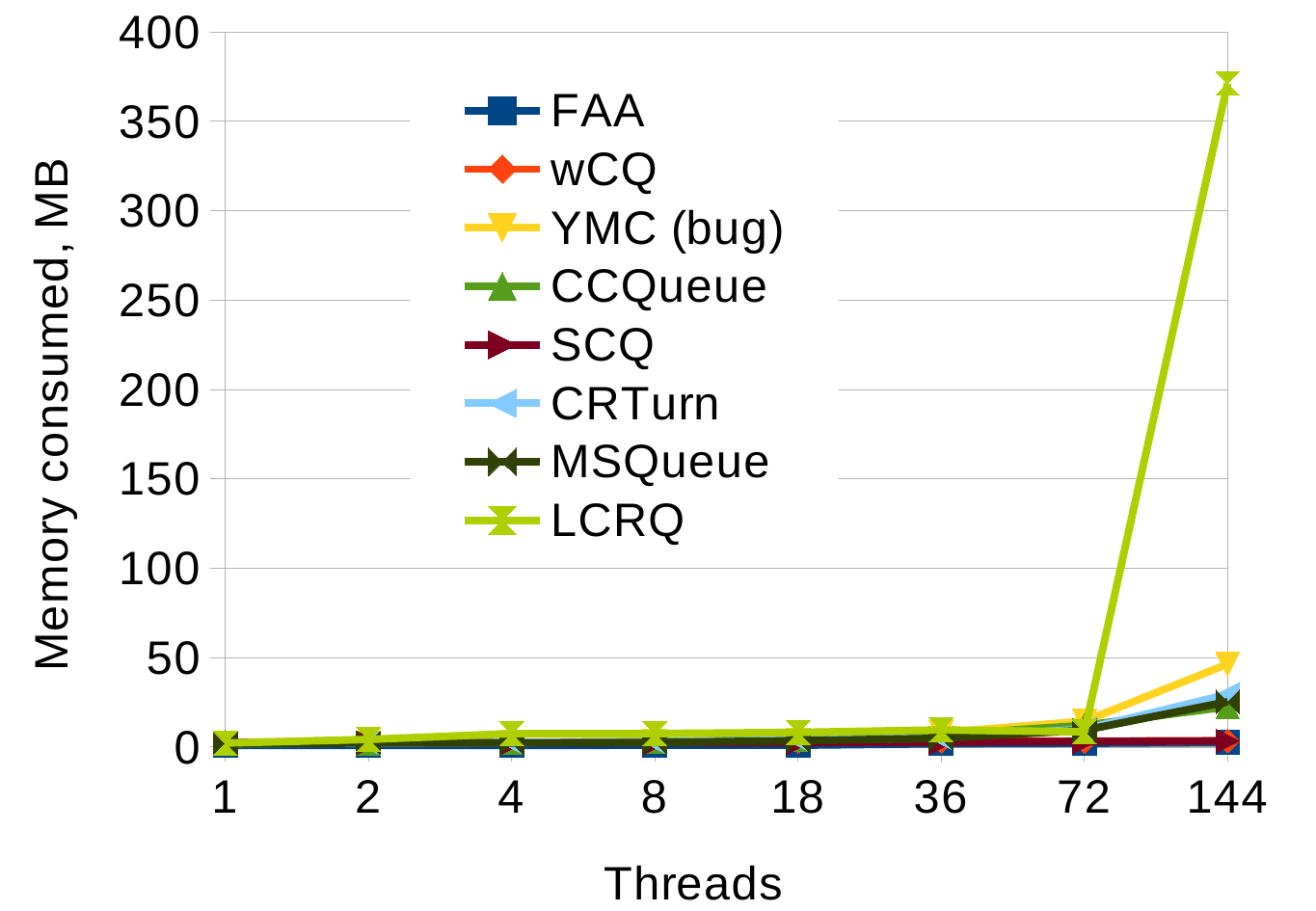}
\caption{Memory usage}
\label{fig:halfhalfmem}
\end{subfigure}%
\begin{subfigure}{.33\textwidth}
\includegraphics[width=\textwidth]{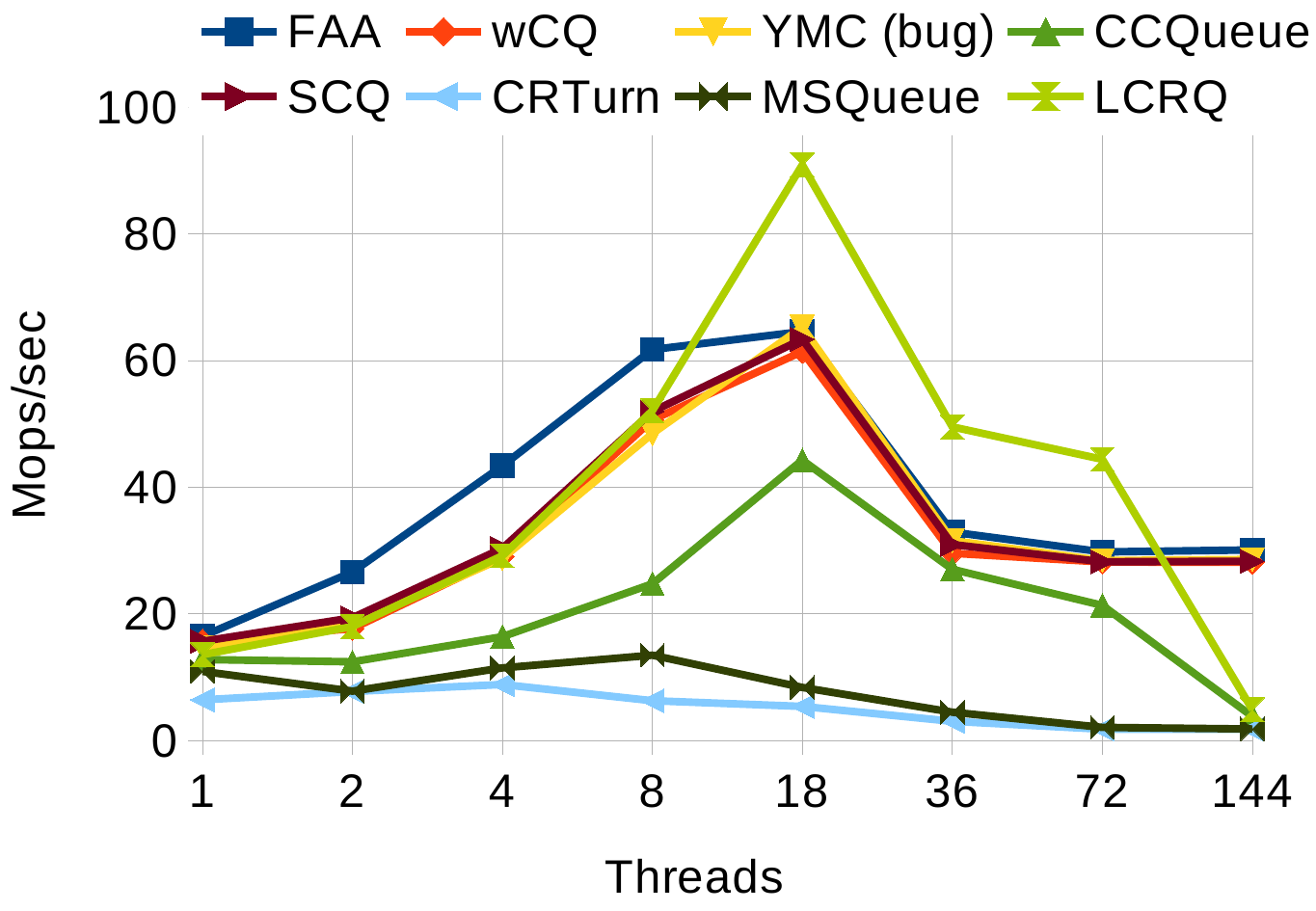}
\caption{Throughput}
\label{fig:halfhalfdel}
\end{subfigure}%
\vspace{-5pt}
\caption{Memory test, x86-64 (Intel Xeon) architecture.}
 \label{fig:memory}
\end{figure*}

We compare wCQ against state-of-the-art algorithms:
\begin{itemize}
\item MSQueue~\cite{Michael:1998:NAP:292022.292026}, a well-known Michael \& Scott's lock-free
queue which is not very performant.
\item LCRQ~\cite{Morrison:2013:FCQ:2442516.2442527}, a queue
that maintains a lock-free list of ring buffers. Ring buffers (CRQs)
are livelock-prone and cannot be used alone.
\item SCQ lock-free queue~\cite{nikolaev:LIPIcs:2019:11335}, a design which
is similar to CRQ but is more memory efficient and avoids livelocks in
ring buffers directly. wCQ is based on SCQ.
\item YMC (Yang \& Mellor-Crummey's) wait-free queue~\cite{Yang:2016:WQF:2851141.2851168}. YMC is flawed (see the discussion in~\cite{pedroWFQUEUEFULL}) in its memory reclamation approach which, strictly described, forfeits wait-freedom.
Since YMC uses F\&A, it is directly relevant in our comparison against wCQ.
\item CRTurn wait-free queue~\cite{pedroWFQUEUEFULL,pedroWFQUEUE}. This is
a truly wait-free queue but as MSQueue, it is not very performant.
\item CCQueue~\cite{Fatourou:2012:RCS:2145816.2145849} is a combining queue, which is not lock-free but still achieves relatively good performance.
\item FAA (fetch-and-add), which is not a true queue algorithm;
it simply atomically increments \textit{Head} and \textit{Tail} when
calling \textit{Dequeue} and \textit{Enqueue} respectively. FAA is only
shown to provide a theoretical performance ``upper bound'' for
F\&A-based queues.
\end{itemize}

SimQueue~\cite{SIMQUEUE} is another wait-free queue, but it
lacks memory reclamation and does not scale as well as the fastest queues. Consequently, SimQueue is not presented here.

For wCQ, we set MAX\_PATIENCE to 16 for \textit{Enqueue} and 64 for \textit{Dequeue}, which
results in taking the slow path relatively infrequently.

To avoid cluttering, we do not separately present a combination of
SCQ with MSQueue or wCQ with a slower wait-free unbounded queue
since: (1) they \textit{only} need that combination when
an unbounded queue is desired (which is unlike LCRQ or YMC, where an outer
layer is always required for progress guarantees) and (2) the
costs of these unbounded queues were measured and found to be
largely dominated by the underlying SCQ and wCQ implementations respectively.
Memory reclamation overheads for queues are generally not that significant,
especially when using more recent fast lock- and wait-free approaches~\cite{10.1145/3178487.3178488,10.1145/3087556.3087588,10.1145/3453483.3454090,10.1145/3293611.3331575,nikolaev_et_al:LIPIcs.DISC.2021.60,10.1145/3332466.3374540}. 
As in~\cite{nikolaev:LIPIcs:2019:11335, Yang:2016:WQF:2851141.2851168},
we use customized reclamation for YMC and hazard pointers
elsewhere (LCRQ, MSQueue, CRTurn). 

\begin{figure*}[ht]
\begin{subfigure}{.34\textwidth}
\includegraphics[width=\textwidth]{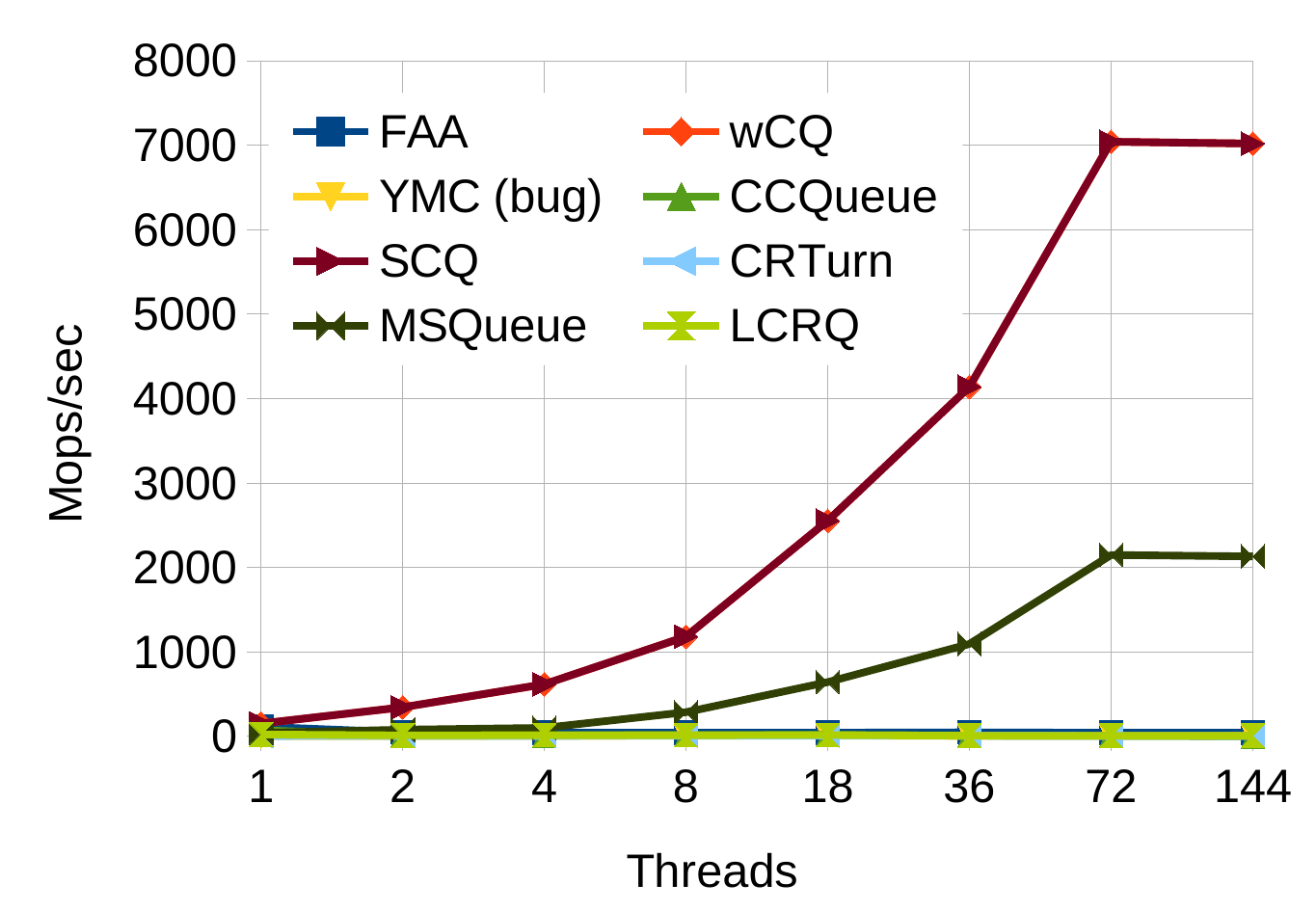}
\caption{Empty \textit{Dequeue} throughput}
\label{fig:emptyx86}
\end{subfigure}%
\hspace{-1em}
\begin{subfigure}{.34\textwidth}
\includegraphics[width=\textwidth]{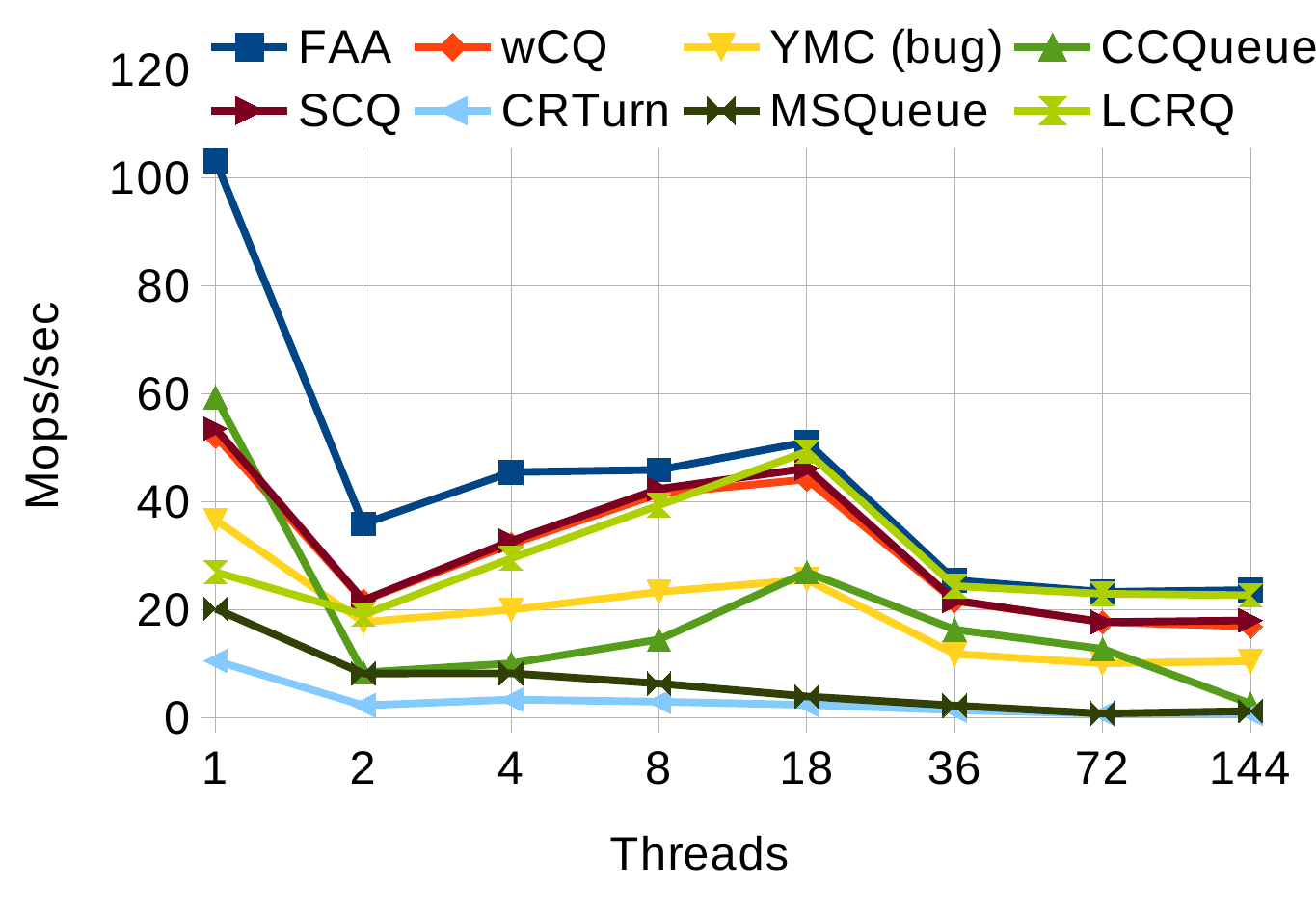}
\caption{Pairwise \textit{Enqueue}-\textit{Dequeue}}
\label{fig:bpairwise}
\end{subfigure}%
\hspace{-1em}
\begin{subfigure}{.34\textwidth}
\includegraphics[width=\textwidth]{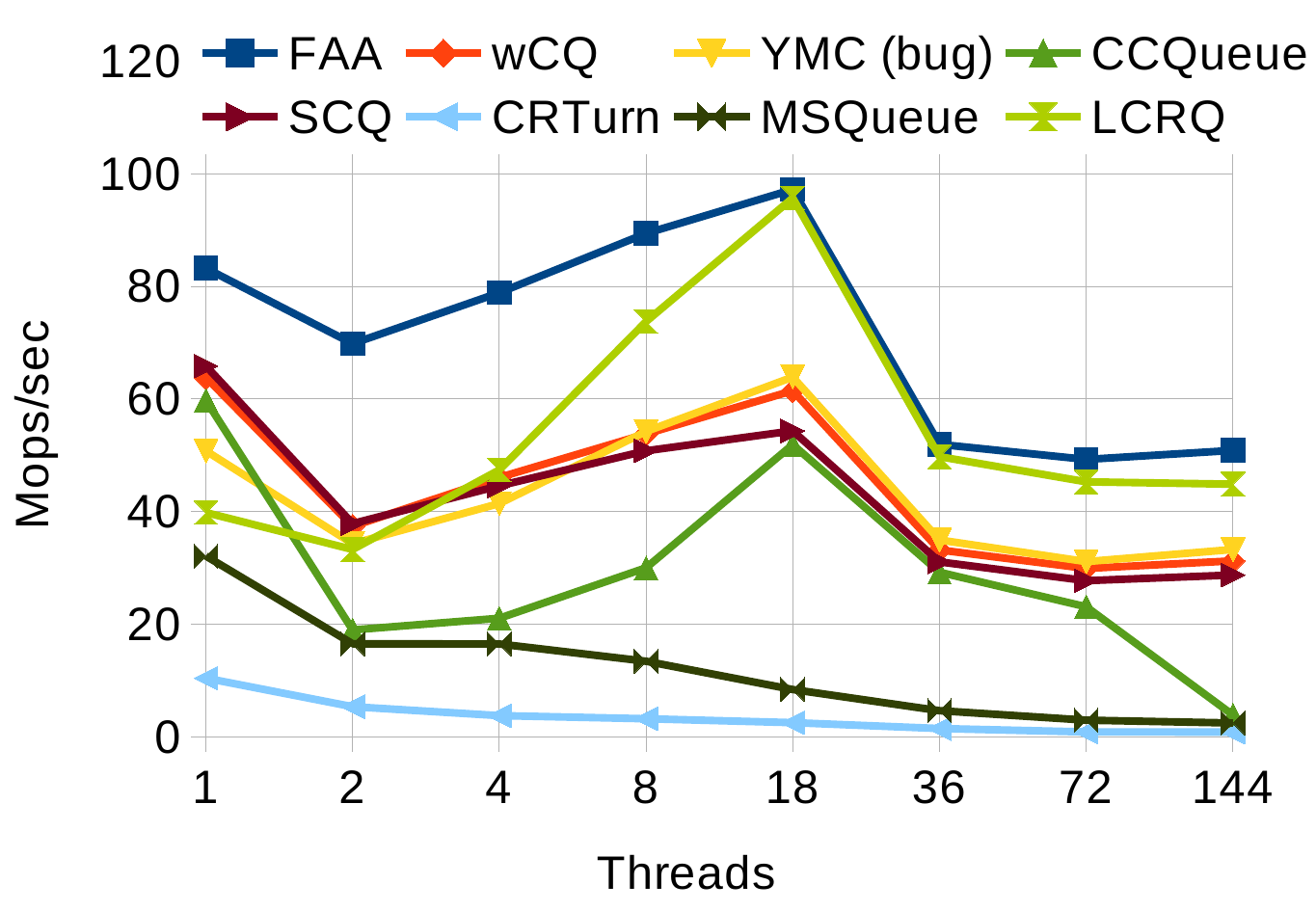}
\caption{50\%/50\% \textit{Enqueue}-\textit{Dequeue}}
\label{fig:bhalfhalf}
\end{subfigure}%
\vspace{-5pt}
\caption{x86-64 (Intel Xeon) architecture.}
\end{figure*}

\begin{figure*}[ht]
\begin{subfigure}{.34\textwidth}
\includegraphics[width=\textwidth]{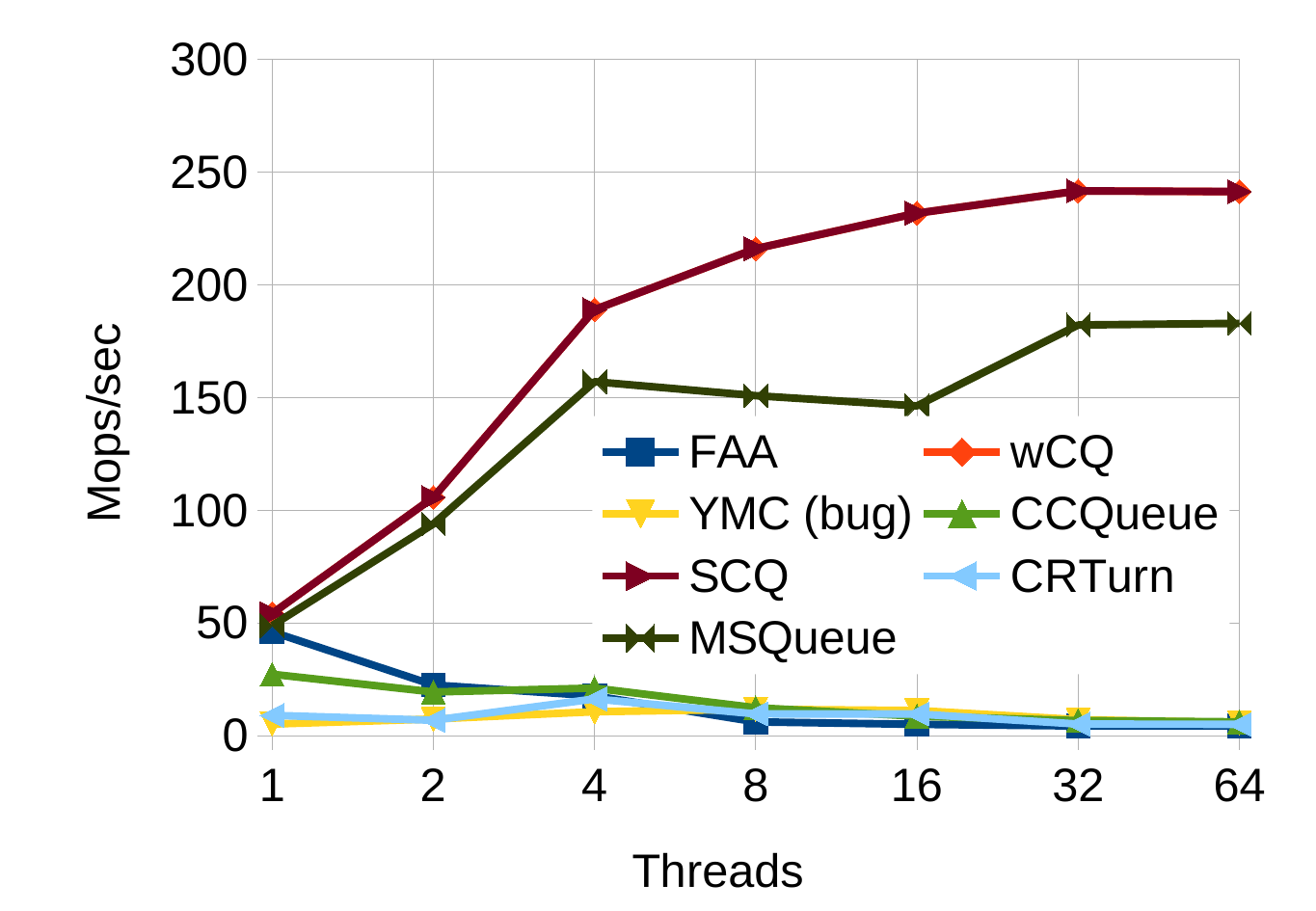}
\caption{Empty \textit{Dequeue} throughput}
\label{fig:emptyppc}
\end{subfigure}%
\hspace{-1em}
\begin{subfigure}{.34\textwidth}
\includegraphics[width=\textwidth]{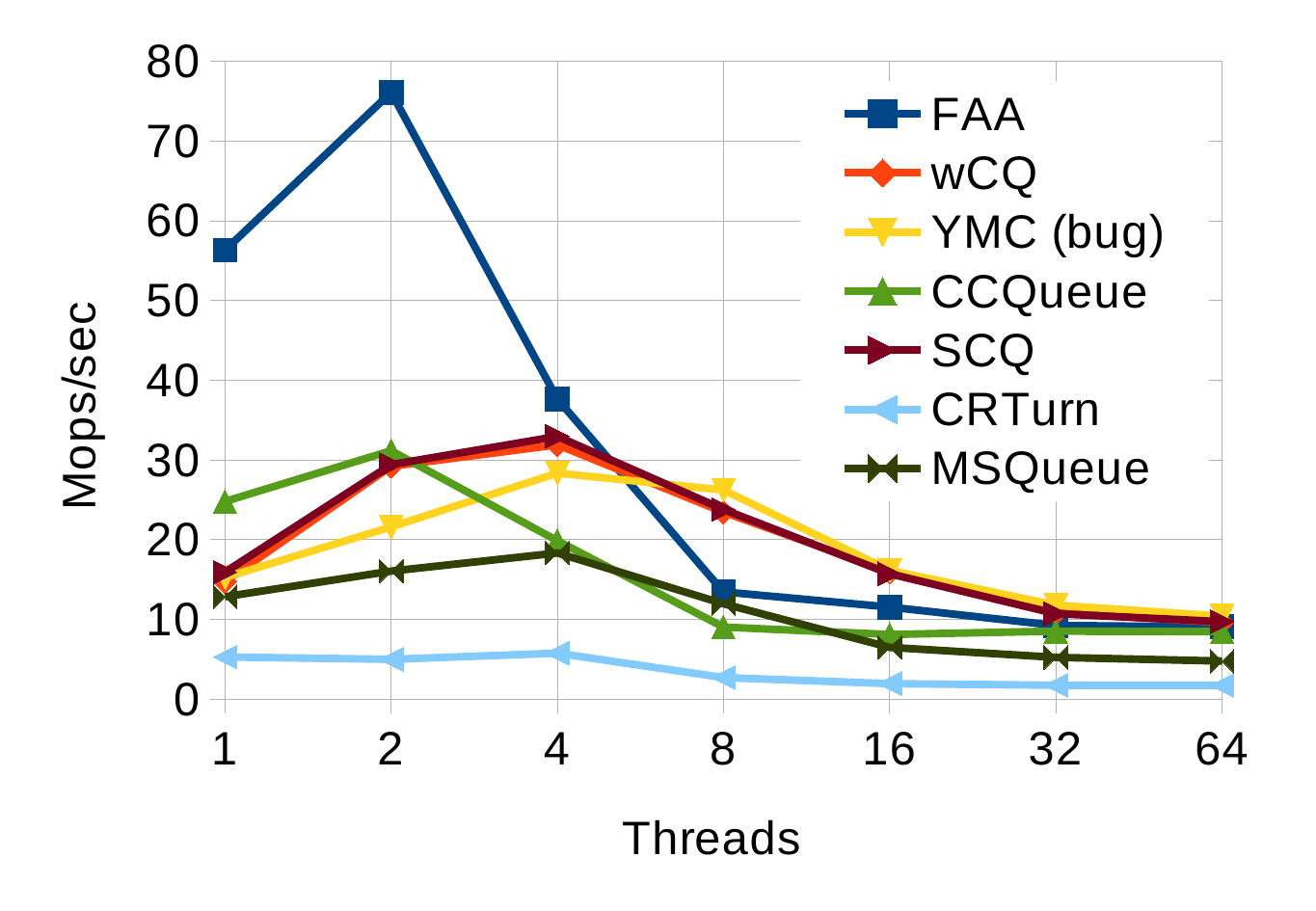}
\caption{Pairwise \textit{Enqueue}-\textit{Dequeue}}
\label{fig:apairwiseppc}
\end{subfigure}%
\hspace{-1em}
\begin{subfigure}{.34\textwidth}
\includegraphics[width=\textwidth]{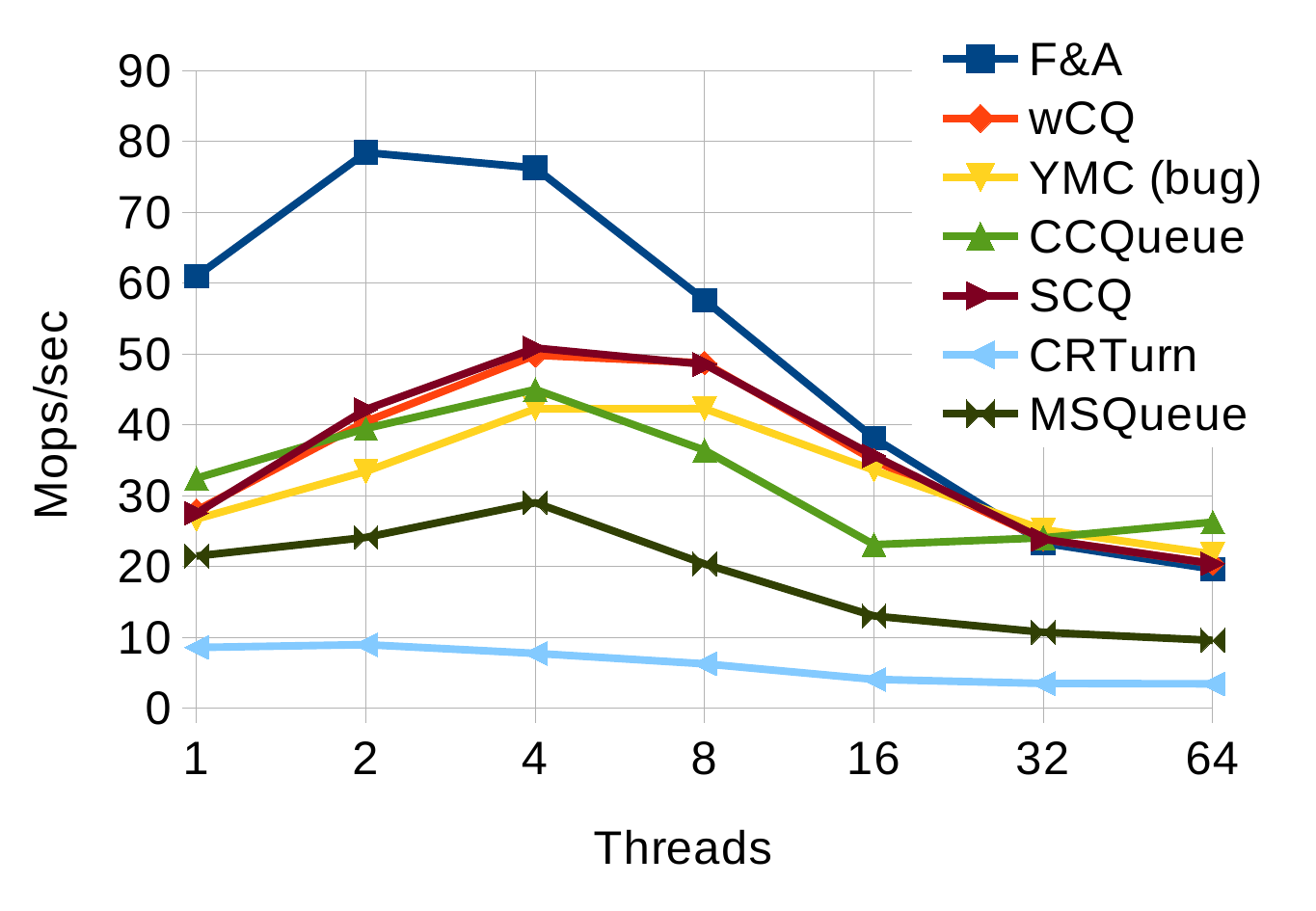}
\caption{50\%/50\% \textit{Enqueue}-\textit{Dequeue}}
\label{fig:ahalfhalfppc}
\end{subfigure}%
\vspace{-5pt}
\caption{PowerPC architecture.}
\end{figure*}

Since each queue has different properties and trade-offs, care must be taken
when assessing their performance. For example, LCRQ can achieve
higher performance in certain tests but it is not portable (e.g., cannot be implemented on PowerPC). Moreover, LCRQ is not as memory efficient as
other queues such as SCQ. Likewise, YMC may exhibit performance similar to SCQ
or wCQ, but its memory reclamation flaws need to be factored in as well.

We performed all tests on x86\_64 and PowerPC machines. The x86\_64 machine has
128GB of RAM and four Intel Xeon E7-8880~v3 (2.30GHz) processors, each with 18 cores with hyper-threading disabled. We used Ubuntu 18.04~LTS installation with gcc 8.4 (-O3 optimization).
wCQ for x86\_64 benefits from CAS2 and hardware-based F\&A.
The PowerPC machine has 64GB of RAM, 64 logical cores (each core has 8 threads), and runs at 3.0 Ghz (8MB~L3 cache). We also used Ubuntu 18.04~LTS installation with gcc 8.4 (-O3 optimization). 
wCQ for PowerPC does not benefit from native F\&A.
Since CAS2 is unavailable on PowerPC, wCQ is implemented via
LL/SC (see Section~\ref{sec:llsc}).
LCRQ requires true CAS2, and its results are only presented for x86\_64.
On both x86\_64 and PowerPC, we use
jemalloc~\cite{jemalloc} due to its better performance~\cite{allocators}.

The benchmark measures average throughput and protects again\-st
outliers. Each point is measured 10 times for
10,000,000 operations in a loop.
The coefficient of variation, as reported by the benchmark, is small
(< 0.01).
We use a relatively small ring buffer ($2^{16}$ entries) for wCQ and SCQ.
For all other queues, we use the default parameters from~\cite{Yang:2016:WQF:2851141.2851168} which appear to be optimal.
Since contention across sockets is expensive, x86-64's throughput peaks for 18 threads
(all 18 threads can fit just one physical CPU).
Likewise, PowerPC's throughput peaks for 4-8 threads.

Our setup for all tests  (except memory efficiency) is slightly
different from~\cite{Yang:2016:WQF:2851141.2851168}, where authors always
used tiny random delays between operations. This resulted in YMC's superior
(relative) performance compared to LCRQ. However, it is unclear if such
a setup reflects any specific practical scenario. Additionally, it
degrades overall absolute throughput, whereas raw performance unhindered by
any delays is often preferred. In our tests, contrary to~\cite{Yang:2016:WQF:2851141.2851168}, LCRQ is often on par or even
superior to YMC, as it is also the case in~\cite{nikolaev:LIPIcs:2019:11335}.

As in~\cite{nikolaev:LIPIcs:2019:11335}, we first measured memory efficiency
using standard malloc (this is merely to make sure that memory pages are
unmapped more frequently, otherwise no fundamental difference in trends exists with jemalloc). This test additionally places tiny random delays between
\textit{Dequeue} and \textit{Enqueue} operations, which we empirically
found to be helpful in amplifying memory efficiency artifacts.
Operations are also chosen randomly,
\textit{Enqueue} for one half of the time, and \textit{Dequeue} for the other
half of the time. LCRQ appears to provide higher overall throughput, but
its memory consumption grows very fast ($\approx380$MB when the number of
threads reaches 144). YMC's memory consumption also grows but slower ($\approx50$MB). Finally, just like SCQ, wCQ is superior in this test. It only needs $\approx1$MB for the ring buffer as well as a negligible amount of memory for per-thread states. The reason for LCRQ's high memory consumption is discussed in~\cite{nikolaev:LIPIcs:2019:11335}. The underlying ring buffers in these algorithms are livelock-prone and can prematurely be ``closed.'' Each ring buffer, for better performance, needs to have at least $2^{12}$ entries, resulting in memory waste in certain cases.

Since performance on empty queues is known to vary~\cite{nikolaev:LIPIcs:2019:11335,Yang:2016:WQF:2851141.2851168}, we also measured \textit{Dequeue} in a tight loop for an empty queue (Figures~\ref{fig:emptyx86}~and~\ref{fig:emptyppc}). wCQ and SCQ both have superior performance due to their use of \textit{Threshold} values. MSQueue also performs well, whereas other queues have inferior performance. In this test, FAA performs poorly since it still incurs cache
invalidations due to RMW operations.

We then measured pairwise \textit{Enqueue}-\textit{Dequeue} operations by
executing \textit{Enqueue} followed by \textit{Dequeue} in a tight
loop. wCQ, SCQ, and LCRQ have superior performance for x86-64 (Figure~\ref{fig:bpairwise}).
wCQ and SCQ both have superior performance for PowerPC (Figure~\ref{fig:apairwiseppc}). In general,
no significant performance difference is observed for wCQ vs. SCQ.
YMC and other queues have worse performance.

We repeated the same experiment, but this time we selected \textit{Enqueue} and \textit{Dequeue} randomly. We executed \textit{Enqueue} 50\% of the time and
\textit{Dequeue} 50\% of the time.
wCQ's, SCQ's, and YMC's performance is very similar for x86-64 (Figure~\ref{fig:bhalfhalf}).
wCQ outperforms SCQ slightly due to reduced cache contention since wCQ's entries are larger. LCRQ typically outperforms all algorithms but
can be vulnerable to increased memory consumption (Figure~\ref{fig:halfhalfmem} shows this scenario). Moreover, LCRQ is only lock-free.
The remaining queues have inferior performance. wCQ and SCQ outperform
YMC and other queues for PowerPC (Figure~\ref{fig:ahalfhalfppc}).

Overall, wCQ is the fastest wait-free queue. Its performance mostly matches the original SCQ algorithm, but wCQ enables wait-free progress. wCQ generally outperforms YMC, for which memory usage can be unbounded. LCRQ can sometimes yield better performance but it is only lock-free. LCRQ also sometimes suffers from high memory usage.

\section{Related Work}

The literature presents a number of lock- and wait-free queues.
Michael \& Scott's FIFO queue~\cite{Michael:1998:NAP:292022.292026} is a lock-free queue that keeps elements in a list of nodes. CAS is used to delete or insert new elements. Kogan \& Petrank's wait-free queue~\cite{kpWFQUEUE} targets languages with garbage collectors such as Java. CRTurn~\cite{pedroWFQUEUE} is a wait-free queue with built-in memory reclamation. Typically, none of these queues are high-performant. When feasible, certain batching lock-free approaches~\cite{10.1145/3210377.3210388} can be used to alleviate the cost.

Though SimQueue~\cite{SIMQUEUE} uses F\&A and outperforms Michael \& Scott's
FIFO queue by aggregating concurrent operations with one CAS, it is still not
that performant.

Ring buffers can be used for bounded queues. Typical ring buffers require CAS as in~\cite{10.1145/3332466.3374508,Tsigas:2001:SFS:378580.378611}. Unfortunately, CAS scales poorly as the contention grows.
Moreover, a number of approaches~\cite{Freudenthal:1991:PCF:106972.106998,Krizhanovsky:2013:LMM:2492102.2492106,vyakov} are either not linearizable or not lock-free despite claims to the contrary, as discussed in~\cite{Feldman:2015:WMM:2835260.2835264,ringdisapp,Morrison:2013:FCQ:2442516.2442527}. Thus, designing high-performant, lock-free ring buffers remained a challenging problem.

Recently, a number of fast lock-free queues have been developed, which are inspired by ring buffers. LCRQ~\cite{Morrison:2013:FCQ:2442516.2442527} implements high-performant, livelock-prone
ring buffers that use F\&A. To workaround livelocks, LCRQ links ``stalled'' ring buffers to Michael \& Scott's lock-free list of ring buffers. This, however,
may result in poor memory efficiency. SCQ~\cite{nikolaev:LIPIcs:2019:11335}, inspired by LCRQ, goes a step further by implementing a high-performant, lock-free ring buffer which can also be linked to Michael \& Scott's lock-free list of ring buffers to create more memory efficient unbounded queues. Our proposed wCQ algorithm, inspired by SCQ, takes another step to devise a fully wait-free ring buffer. A slower queue (e.g., CRTurn) can be used as an outer layer for wCQ to
implement unbounded queues. YMC~\cite{Yang:2016:WQF:2851141.2851168}, also partially inspired by LCRQ, attempted to create a wait-free queue. However, as discussed in~\cite{pedroWFQUEUEFULL}, YMC is flawed in its memory reclamation approach, and is therefore not truly wait-free. Finally, LOO~\cite{FIFOTPDS}, another recent queue inspired by LCRQ, implements a specialized memory reclamation approach but is only lock-free.

Garbage collectors can alleviate reclamation problems but are not
always practical, e.g., in C/C++. Furthermore, to our knowledge, there
is no wait-free garbage collector with \emph{limited} overheads and \emph{bounded} memory usage. FreeAccess~\cite{Cohen:2018:DSD:3288538.3276513} and OrcGC~\cite{OrcGC} are efficient but they are only lock-free.

Finally, memory-boundness is an important goal in general, as evidenced
by recent works. wCQ is close to the theoretical bound discussed in~\cite{MEMORYOPT}.

\section{Conclusion}

We presented wCQ, a fast wait-free FIFO queue.
wCQ is the \emph{first} high-performant wait-free queue for which memory usage
is bound\-ed. Prior approaches, such as YMC~\cite{Yang:2016:WQF:2851141.2851168}, also aimed
at high performance, but failed to guarantee bounded memory usage.
Strictly described, YMC is not wait-free as it is
blocking when memory is exhausted.

Similar to SCQ's original lock-free design, wCQ uses F\&A (fetch-and-add) for the
most contended hot spots of the algorithm:
\textit{Head} and \textit{Tail} pointers.
Although native F\&A is recommended, wCQ retains wait-freedom
even on architectures without F\&A, such as PowerPC or MIPS. wCQ requires
double-width CAS to implement slow-path procedures correctly but can also
be implemented on architectures with ordinary LL/SC (PowerPC or MIPS).
Thus, wCQ largely retains SCQ's portability across different architectures.

Though Kogan-Petrank's method can be used to create
wait-free queues with CAS~\cite{kpWFQUEUE}, wCQ addresses unique challenges since it avoids
dynamic allocation and uses F\&A. We hope that wCQ will spur
further research in creating \textit{better performing} wait-free data structures with F\&A.

Unbounded queues  can be created by linking wCQs together,
similarly to LCRQ or LSCQ, which use Michael \& Scott's lock-free queue as an outer layer.
The underlying algorithm need not be performant
since new circular queues are only allocated very infrequently. Assuming that
the underlying (non-performant) algorithm is truly wait-free with bounded memory
usage such as CRTurn~\cite{pedroWFQUEUEFULL, pedroWFQUEUE}, wCQ's complete benefits are retained even for unbounded queues.

Our slow\_F\&A idea can be adopted in other data structures that
rely on F\&A for improved performance.

\section*{Availability}
We provide wCQ's code at \url{https://github.com/rusnikola/wfqueue}.

An extended and up-to-date arXiv version of the paper is available at \url{https://arxiv.org/abs/2201.02179}.

\section*{Acknowledgments}

We would like to thank the anonymous reviewers for their invaluable feedback.

A preliminary version of the algorithm previously appeared as a poster at PPoPP '22~\cite{10.1145/3503221.3508440}.

This work is supported in part by AFOSR under grant FA9550-16-1-0371 and ONR under grants N00014-18-1-2022 and N00014-19-1-2493.

\bibliography{lockfree}


\begin{thebibliography}{45}


\ifx \showCODEN    \undefined \def \showCODEN     #1{\unskip}     \fi
\ifx \showDOI      \undefined \def \showDOI       #1{#1}\fi
\ifx \showISBNx    \undefined \def \showISBNx     #1{\unskip}     \fi
\ifx \showISBNxiii \undefined \def \showISBNxiii  #1{\unskip}     \fi
\ifx \showISSN     \undefined \def \showISSN      #1{\unskip}     \fi
\ifx \showLCCN     \undefined \def \showLCCN      #1{\unskip}     \fi
\ifx \shownote     \undefined \def \shownote      #1{#1}          \fi
\ifx \showarticletitle \undefined \def \showarticletitle #1{#1}   \fi
\ifx \showURL      \undefined \def \showURL       {\relax}        \fi
\providecommand\bibfield[2]{#2}
\providecommand\bibinfo[2]{#2}
\providecommand\natexlab[1]{#1}
\providecommand\showeprint[2][]{arXiv:#2}

\bibitem[Aksenov et~al\mbox{.}(2022)]%
        {MEMORYOPT}
\bibfield{author}{\bibinfo{person}{Vitaly Aksenov}, \bibinfo{person}{Nikita
  Koval}, {and} \bibinfo{person}{Petr Kuznetsov}.}
  \bibinfo{year}{2022}\natexlab{}.
\newblock \showarticletitle{{Memory-Optimality for Non-Blocking Containers}}.
\newblock
\urldef\tempurl%
\url{https://arxiv.org/abs/2104.15003}
\showURL{%
\tempurl}


\bibitem[{ARM}(2022)]%
        {arm:manual}
\bibfield{author}{\bibinfo{person}{{ARM}}.} \bibinfo{year}{2022}\natexlab{}.
\newblock \bibinfo{title}{{ARM Architecture Reference Manual}}.
\newblock \bibinfo{howpublished}{\url{http://developer.arm.com/}}.
\newblock


\bibitem[Cohen(2018)]%
        {Cohen:2018:DSD:3288538.3276513}
\bibfield{author}{\bibinfo{person}{Nachshon Cohen}.}
  \bibinfo{year}{2018}\natexlab{}.
\newblock \showarticletitle{{Every Data Structure Deserves Lock-free Memory
  Reclamation}}.
\newblock \bibinfo{journal}{\emph{Proc. ACM Program. Lang.}}
  \bibinfo{volume}{2}, \bibinfo{number}{OOPSLA}, Article
  \bibinfo{articleno}{143} (\bibinfo{date}{Oct.} \bibinfo{year}{2018}),
  \bibinfo{numpages}{24}~pages.
\newblock
\showISSN{2475-1421}
\urldef\tempurl%
\url{https://doi.org/10.1145/3276513}
\showDOI{\tempurl}


\bibitem[Correia et~al\mbox{.}(2021)]%
        {OrcGC}
\bibfield{author}{\bibinfo{person}{Andreia Correia}, \bibinfo{person}{Pedro
  Ramalhete}, {and} \bibinfo{person}{Pascal Felber}.}
  \bibinfo{year}{2021}\natexlab{}.
\newblock \showarticletitle{{OrcGC: Automatic Lock-Free Memory Reclamation}}.
  In \bibinfo{booktitle}{\emph{Proceedings of the 26th ACM SIGPLAN Symposium on
  Principles and Practice of Parallel Programming}}
  \emph{(\bibinfo{series}{PPoPP '21})}. \bibinfo{publisher}{ACM},
  \bibinfo{pages}{205--218}.
\newblock
\showISBNx{9781450382946}
\urldef\tempurl%
\url{https://doi.org/10.1145/3437801.3441596}
\showDOI{\tempurl}


\bibitem[{DPDK Developers}(2022)]%
        {DPDK}
\bibfield{author}{\bibinfo{person}{{DPDK Developers}}.}
  \bibinfo{year}{2022}\natexlab{}.
\newblock \showarticletitle{{Data Plane Development Kit (DPDK)}}.
\newblock
\newblock
\shownote{\url{https://dpdk.org/}}.


\bibitem[Evans(2006)]%
        {jemalloc}
\bibfield{author}{\bibinfo{person}{Jason Evans}.}
  \bibinfo{year}{2006}\natexlab{}.
\newblock \showarticletitle{{A scalable concurrent malloc(3) implementation for
  FreeBSD}}. In \bibinfo{booktitle}{\emph{Proceedings of the BSDCan Conference,
  Ottawa, Canada}}.
\newblock
\urldef\tempurl%
\url{https://www.bsdcan.org/2006/papers/jemalloc.pdf}
\showURL{%
\tempurl}


\bibitem[Fatourou and Kallimanis(2011)]%
        {SIMQUEUE}
\bibfield{author}{\bibinfo{person}{Panagiota Fatourou} {and}
  \bibinfo{person}{Nikolaos~D. Kallimanis}.} \bibinfo{year}{2011}\natexlab{}.
\newblock \showarticletitle{{A Highly-Efficient Wait-Free Universal
  Construction}}. In \bibinfo{booktitle}{\emph{Proceedings of the 23rd Annual
  ACM Symposium on Parallelism in Algorithms and Architectures}} (San Jose,
  California, USA) \emph{(\bibinfo{series}{SPAA '11})}.
  \bibinfo{publisher}{ACM}, \bibinfo{address}{New York, NY, USA},
  \bibinfo{pages}{325--334}.
\newblock
\showISBNx{9781450307437}
\urldef\tempurl%
\url{https://doi.org/10.1145/1989493.1989549}
\showDOI{\tempurl}


\bibitem[Fatourou and Kallimanis(2012)]%
        {Fatourou:2012:RCS:2145816.2145849}
\bibfield{author}{\bibinfo{person}{Panagiota Fatourou} {and}
  \bibinfo{person}{Nikolaos~D. Kallimanis}.} \bibinfo{year}{2012}\natexlab{}.
\newblock \showarticletitle{{Revisiting the Combining Synchronization
  Technique}}. In \bibinfo{booktitle}{\emph{Proceedings of the 17th ACM SIGPLAN
  Symposium on Principles and Practice of Parallel Programming}} (New Orleans,
  Louisiana, USA) \emph{(\bibinfo{series}{PPoPP '12})}.
  \bibinfo{publisher}{ACM}, \bibinfo{address}{New York, NY, USA},
  \bibinfo{pages}{257--266}.
\newblock
\showISBNx{978-1-4503-1160-1}
\urldef\tempurl%
\url{https://doi.org/10.1145/2145816.2145849}
\showDOI{\tempurl}


\bibitem[Feldman and Dechev(2015)]%
        {Feldman:2015:WMM:2835260.2835264}
\bibfield{author}{\bibinfo{person}{Steven Feldman} {and}
  \bibinfo{person}{Damian Dechev}.} \bibinfo{year}{2015}\natexlab{}.
\newblock \showarticletitle{{A Wait-free Multi-producer Multi-consumer Ring
  Buffer}}.
\newblock \bibinfo{journal}{\emph{ACM SIGAPP Applied Computing Review}}
  \bibinfo{volume}{15}, \bibinfo{number}{3} (\bibinfo{date}{Oct.}
  \bibinfo{year}{2015}), \bibinfo{pages}{59--71}.
\newblock
\showISSN{1559-6915}
\urldef\tempurl%
\url{https://doi.org/10.1145/2835260.2835264}
\showDOI{\tempurl}


\bibitem[Freudenthal and Gottlieb(1991)]%
        {Freudenthal:1991:PCF:106972.106998}
\bibfield{author}{\bibinfo{person}{Eric Freudenthal} {and}
  \bibinfo{person}{Allan Gottlieb}.} \bibinfo{year}{1991}\natexlab{}.
\newblock \showarticletitle{{Process Coordination with Fetch-and-increment}}.
  In \bibinfo{booktitle}{\emph{Proceedings of the 4th International Conference
  on Architectural Support for Programming Languages and Operating Systems}}
  (Santa Clara, California, USA) \emph{(\bibinfo{series}{ASPLOS IV})}.
  \bibinfo{pages}{260--268}.
\newblock
\showISBNx{0-89791-380-9}
\urldef\tempurl%
\url{https://doi.org/10.1145/106972.106998}
\showDOI{\tempurl}


\bibitem[Giersch and Nolte(2022)]%
        {FIFOTPDS}
\bibfield{author}{\bibinfo{person}{Oliver Giersch} {and} \bibinfo{person}{Jörg
  Nolte}.} \bibinfo{year}{2022}\natexlab{}.
\newblock \showarticletitle{{Fast and Portable Concurrent FIFO Queues With
  Deterministic Memory Reclamation}}.
\newblock \bibinfo{journal}{\emph{IEEE Trans. on Parallel and Distributed
  Systems}} \bibinfo{volume}{33}, \bibinfo{number}{3} (\bibinfo{year}{2022}),
  \bibinfo{pages}{604--616}.
\newblock
\urldef\tempurl%
\url{https://doi.org/10.1109/TPDS.2021.3097901}
\showDOI{\tempurl}


\bibitem[Hendler et~al\mbox{.}(2004)]%
        {Hendler:2004:SLS:1007912.1007944}
\bibfield{author}{\bibinfo{person}{Danny Hendler}, \bibinfo{person}{Nir
  Shavit}, {and} \bibinfo{person}{Lena Yerushalmi}.}
  \bibinfo{year}{2004}\natexlab{}.
\newblock \showarticletitle{{A Scalable Lock-free Stack Algorithm}}. In
  \bibinfo{booktitle}{\emph{Proceedings of the 16th ACM Symposium on
  Parallelism in Algorithms and Architectures}} (Barcelona, Spain)
  \emph{(\bibinfo{series}{SPAA '04})}. \bibinfo{publisher}{ACM},
  \bibinfo{address}{New York, NY, USA}, \bibinfo{pages}{206--215}.
\newblock
\showISBNx{1-58113-840-7}
\urldef\tempurl%
\url{https://doi.org/10.1145/1007912.1007944}
\showDOI{\tempurl}


\bibitem[Herlihy and Wing(1990)]%
        {10.1145/78969.78972}
\bibfield{author}{\bibinfo{person}{Maurice~P. Herlihy} {and}
  \bibinfo{person}{Jeannette~M. Wing}.} \bibinfo{year}{1990}\natexlab{}.
\newblock \showarticletitle{{Linearizability: A Correctness Condition for
  Concurrent Objects}}.
\newblock \bibinfo{journal}{\emph{ACM Trans. Program. Lang. Syst.}}
  \bibinfo{volume}{12}, \bibinfo{number}{3} (\bibinfo{date}{jul}
  \bibinfo{year}{1990}), \bibinfo{pages}{463--492}.
\newblock
\showISSN{0164-0925}
\urldef\tempurl%
\url{https://doi.org/10.1145/78969.78972}
\showDOI{\tempurl}


\bibitem[{IBM}(2005)]%
        {ppc:manual}
\bibfield{author}{\bibinfo{person}{{IBM}}.} \bibinfo{year}{2005}\natexlab{}.
\newblock \bibinfo{title}{{PowerPC Architecture Book, Version 2.02. Book I:
  PowerPC User Instruction Set Architecture}}.
\newblock \bibinfo{howpublished}{\url{http://www.ibm.com/developerworks/}}.
\newblock


\bibitem[{Intel}(2022)]%
        {intel:manual}
\bibfield{author}{\bibinfo{person}{{Intel}}.} \bibinfo{year}{2022}\natexlab{}.
\newblock \bibinfo{title}{{Intel 64 and IA-32 Architectures Developer's
  Manual}}.
\newblock \bibinfo{howpublished}{\url{http://www.intel.com/}}.
\newblock


\bibitem[Kappes and Anastasiadis(2021)]%
        {10.1145/3437801.3441583}
\bibfield{author}{\bibinfo{person}{Giorgos Kappes} {and}
  \bibinfo{person}{Stergios~V. Anastasiadis}.} \bibinfo{year}{2021}\natexlab{}.
\newblock \showarticletitle{{POSTER: A Lock-Free Relaxed Concurrent Queue for
  Fast Work Distribution}}. In \bibinfo{booktitle}{\emph{Proceedings of the
  26th ACM SIGPLAN Symposium on Principles and Practice of Parallel
  Programming}} (Virtual Event, Republic of Korea)
  \emph{(\bibinfo{series}{PPoPP '21})}. \bibinfo{publisher}{ACM},
  \bibinfo{address}{New York, NY, USA}, \bibinfo{pages}{454--456}.
\newblock
\showISBNx{9781450382946}
\urldef\tempurl%
\url{https://doi.org/10.1145/3437801.3441583}
\showDOI{\tempurl}


\bibitem[Kirsch et~al\mbox{.}(2013)]%
        {Kirsch:2013:FSL:2960356.2960376}
\bibfield{author}{\bibinfo{person}{Christoph~M. Kirsch},
  \bibinfo{person}{Michael Lippautz}, {and} \bibinfo{person}{Hannes Payer}.}
  \bibinfo{year}{2013}\natexlab{}.
\newblock \showarticletitle{{Fast and Scalable, Lock-Free k-FIFO Queues}}. In
  \bibinfo{booktitle}{\emph{Proceedings of the 12th International Conference on
  Parallel Computing Technologies - Volume 7979}}.
  \bibinfo{publisher}{Springer-Verlag}, \bibinfo{address}{Berlin, Heidelberg},
  \bibinfo{pages}{208--223}.
\newblock
\showISBNx{978-3-642-39957-2}
\urldef\tempurl%
\url{https://doi.org/10.1007/978-3-642-39958-9_18}
\showDOI{\tempurl}


\bibitem[Kogan and Petrank(2011)]%
        {kpWFQUEUE}
\bibfield{author}{\bibinfo{person}{Alex Kogan} {and} \bibinfo{person}{Erez
  Petrank}.} \bibinfo{year}{2011}\natexlab{}.
\newblock \showarticletitle{{Wait-free Queues with Multiple Enqueuers and
  Dequeuers}}. In \bibinfo{booktitle}{\emph{Proceedings of the 16th ACM
  Symposium on Principles and Practice of Parallel Programming}} (San Antonio,
  TX, USA) \emph{(\bibinfo{series}{PPoPP '11})}. \bibinfo{publisher}{ACM},
  \bibinfo{address}{New York, NY, USA}, \bibinfo{pages}{223--234}.
\newblock
\showISBNx{978-1-4503-0119-0}
\urldef\tempurl%
\url{https://doi.org/10.1145/1941553.1941585}
\showDOI{\tempurl}


\bibitem[Kogan and Petrank(2012)]%
        {Kogan:2012:MCF:2145816.2145835}
\bibfield{author}{\bibinfo{person}{Alex Kogan} {and} \bibinfo{person}{Erez
  Petrank}.} \bibinfo{year}{2012}\natexlab{}.
\newblock \showarticletitle{{A Methodology for Creating Fast Wait-free Data
  Structures}}. In \bibinfo{booktitle}{\emph{Proceedings of the 17th ACM
  SIGPLAN Symposium on Principles and Practice of Parallel Programming}} (New
  Orleans, Louisiana, USA) \emph{(\bibinfo{series}{PPoPP '12})}.
  \bibinfo{publisher}{ACM}, \bibinfo{address}{New York, NY, USA},
  \bibinfo{pages}{141--150}.
\newblock
\showISBNx{978-1-4503-1160-1}
\urldef\tempurl%
\url{https://doi.org/10.1145/2145816.2145835}
\showDOI{\tempurl}


\bibitem[Koval and Aksenov(2020)]%
        {10.1145/3332466.3374508}
\bibfield{author}{\bibinfo{person}{Nikita Koval} {and} \bibinfo{person}{Vitaly
  Aksenov}.} \bibinfo{year}{2020}\natexlab{}.
\newblock \showarticletitle{{POSTER: Restricted Memory-Friendly Lock-Free
  Bounded Queues}}. In \bibinfo{booktitle}{\emph{Proceedings of the 25th ACM
  SIGPLAN Symposium on Principles and Practice of Parallel Programming}} (San
  Diego, California) \emph{(\bibinfo{series}{PPoPP '20})}.
  \bibinfo{publisher}{ACM}, \bibinfo{address}{New York, NY, USA},
  \bibinfo{pages}{433--434}.
\newblock
\showISBNx{9781450368186}
\urldef\tempurl%
\url{https://doi.org/10.1145/3332466.3374508}
\showDOI{\tempurl}


\bibitem[Krizhanovsky(2013)]%
        {Krizhanovsky:2013:LMM:2492102.2492106}
\bibfield{author}{\bibinfo{person}{Alexander Krizhanovsky}.}
  \bibinfo{year}{2013}\natexlab{}.
\newblock \showarticletitle{{Lock-free Multi-producer Multi-consumer Queue on
  Ring Buffer}}.
\newblock \bibinfo{journal}{\emph{Linux J.}} \bibinfo{volume}{2013},
  \bibinfo{number}{228}, Article \bibinfo{articleno}{4} (\bibinfo{year}{2013}).
\newblock
\showISSN{1075-3583}


\bibitem[Lamport(1979)]%
        {Lamport:1979:MMC:1311099.1311750}
\bibfield{author}{\bibinfo{person}{Leslie Lamport}.}
  \bibinfo{year}{1979}\natexlab{}.
\newblock \showarticletitle{{How to Make a Multiprocessor Computer That
  Correctly Executes Multiprocess Programs}}.
\newblock \bibinfo{journal}{\emph{IEEE Trans. Comput.}} \bibinfo{volume}{28},
  \bibinfo{number}{9} (\bibinfo{date}{Sept.} \bibinfo{year}{1979}),
  \bibinfo{pages}{690--691}.
\newblock
\showISSN{0018-9340}


\bibitem[{Liblfds}(2022a)]%
        {liblfds}
\bibfield{author}{\bibinfo{person}{{Liblfds}}.}
  \bibinfo{year}{2022}\natexlab{a}.
\newblock \bibinfo{title}{{Lock-free Data Structure Library}}.
\newblock
\newblock
\newblock
\shownote{\url{https://www.liblfds.org}}.


\bibitem[{Liblfds}(2022b)]%
        {ringdisapp}
\bibfield{author}{\bibinfo{person}{{Liblfds}}.}
  \bibinfo{year}{2022}\natexlab{b}.
\newblock \bibinfo{title}{Ringbuffer disappointment 2016-04-29}.
\newblock
\newblock
\newblock
\shownote{\url{https://www.liblfds.org/slblog/index.html}}.


\bibitem[{Lockless Inc.}(2022)]%
        {allocators}
\bibfield{author}{\bibinfo{person}{{Lockless Inc.}}}
  \bibinfo{year}{2022}\natexlab{}.
\newblock \bibinfo{title}{{Memory Allocator Benchmarks}}.
\newblock \bibinfo{howpublished}{\url{https://locklessinc.com}}.
\newblock


\bibitem[Michael and Scott(1998)]%
        {Michael:1998:NAP:292022.292026}
\bibfield{author}{\bibinfo{person}{Maged~M. Michael} {and}
  \bibinfo{person}{Michael~L. Scott}.} \bibinfo{year}{1998}\natexlab{}.
\newblock \showarticletitle{{Nonblocking Algorithms and Preemption-Safe Locking
  on Multiprogrammed Shared Memory Multiprocessors}}.
\newblock \bibinfo{journal}{\emph{J. Parallel and Distrib. Comput.}}
  \bibinfo{volume}{51}, \bibinfo{number}{1} (\bibinfo{date}{May}
  \bibinfo{year}{1998}), \bibinfo{pages}{1--26}.
\newblock
\showISSN{0743-7315}
\urldef\tempurl%
\url{https://doi.org/10.1006/jpdc.1998.1446}
\showDOI{\tempurl}


\bibitem[Milman et~al\mbox{.}(2018)]%
        {10.1145/3210377.3210388}
\bibfield{author}{\bibinfo{person}{Gal Milman}, \bibinfo{person}{Alex Kogan},
  \bibinfo{person}{Yossi Lev}, \bibinfo{person}{Victor Luchangco}, {and}
  \bibinfo{person}{Erez Petrank}.} \bibinfo{year}{2018}\natexlab{}.
\newblock \showarticletitle{{BQ: A Lock-Free Queue with Batching}}. In
  \bibinfo{booktitle}{\emph{Proceedings of the 30th on Symposium on Parallelism
  in Algorithms and Architectures}} (Vienna, Austria)
  \emph{(\bibinfo{series}{SPAA '18})}. \bibinfo{publisher}{ACM},
  \bibinfo{address}{New York, NY, USA}, \bibinfo{pages}{99--109}.
\newblock
\showISBNx{9781450357999}
\urldef\tempurl%
\url{https://doi.org/10.1145/3210377.3210388}
\showDOI{\tempurl}


\bibitem[{MIPS}(2022)]%
        {mips:manual}
\bibfield{author}{\bibinfo{person}{{MIPS}}.} \bibinfo{year}{2022}\natexlab{}.
\newblock \bibinfo{title}{{MIPS32/MIPS64 Rev. 6.06}}.
\newblock
  \bibinfo{howpublished}{\url{http://www.mips.com/products/architectures/}}.
\newblock


\bibitem[Moir et~al\mbox{.}(2005)]%
        {Moir:2005:UEI:1073970.1074013}
\bibfield{author}{\bibinfo{person}{Mark Moir}, \bibinfo{person}{Daniel
  Nussbaum}, \bibinfo{person}{Ori Shalev}, {and} \bibinfo{person}{Nir Shavit}.}
  \bibinfo{year}{2005}\natexlab{}.
\newblock \showarticletitle{{Using Elimination to Implement Scalable and
  Lock-free FIFO Queues}}. In \bibinfo{booktitle}{\emph{Proceedings of the 17th
  ACM Symposium on Parallelism in Algorithms and Architectures}} (Las Vegas,
  Nevada, USA) \emph{(\bibinfo{series}{SPAA '05})}. \bibinfo{pages}{253--262}.
\newblock
\showISBNx{1-58113-986-1}
\urldef\tempurl%
\url{https://doi.org/10.1145/1073970.1074013}
\showDOI{\tempurl}


\bibitem[Morrison and Afek(2013)]%
        {Morrison:2013:FCQ:2442516.2442527}
\bibfield{author}{\bibinfo{person}{Adam Morrison} {and} \bibinfo{person}{Yehuda
  Afek}.} \bibinfo{year}{2013}\natexlab{}.
\newblock \showarticletitle{{Fast Concurrent Queues for x86 Processors}}. In
  \bibinfo{booktitle}{\emph{Proceedings of the 18th ACM SIGPLAN Symposium on
  Principles and Practice of Parallel Programming}} (Shenzhen, China)
  \emph{(\bibinfo{series}{PPoPP '13})}. \bibinfo{publisher}{ACM},
  \bibinfo{address}{New York, NY, USA}, \bibinfo{pages}{103--112}.
\newblock
\showISBNx{978-1-4503-1922-5}
\urldef\tempurl%
\url{https://doi.org/10.1145/2442516.2442527}
\showDOI{\tempurl}


\bibitem[Nikolaev(2019)]%
        {nikolaev:LIPIcs:2019:11335}
\bibfield{author}{\bibinfo{person}{Ruslan Nikolaev}.}
  \bibinfo{year}{2019}\natexlab{}.
\newblock \showarticletitle{{A Scalable, Portable, and Memory-Efficient
  Lock-Free FIFO Queue}}. In \bibinfo{booktitle}{\emph{33rd International
  Symposium on Distributed Computing (DISC 2019)}}
  \emph{(\bibinfo{series}{Leibniz International Proceedings in Informatics
  (LIPIcs)}, Vol.~\bibinfo{volume}{146})},
  \bibfield{editor}{\bibinfo{person}{Jukka Suomela}} (Ed.).
  \bibinfo{publisher}{Schloss Dagstuhl--Leibniz-Zentrum fuer Informatik},
  \bibinfo{address}{Dagstuhl, Germany}, \bibinfo{pages}{28:1--28:16}.
\newblock
\showISBNx{978-3-95977-126-9}
\showISSN{1868-8969}
\urldef\tempurl%
\url{https://doi.org/10.4230/LIPIcs.DISC.2019.28}
\showDOI{\tempurl}


\bibitem[Nikolaev and Ravindran(2019)]%
        {10.1145/3293611.3331575}
\bibfield{author}{\bibinfo{person}{Ruslan Nikolaev} {and}
  \bibinfo{person}{Binoy Ravindran}.} \bibinfo{year}{2019}\natexlab{}.
\newblock \showarticletitle{{Brief Announcement: Hyaline: Fast and Transparent
  Lock-Free Memory Reclamation}}. In \bibinfo{booktitle}{\emph{Proceedings of
  the 2019 ACM Symposium on Principles of Distributed Computing}} (Toronto ON,
  Canada) \emph{(\bibinfo{series}{PODC '19})}. \bibinfo{publisher}{ACM},
  \bibinfo{address}{New York, NY, USA}, \bibinfo{pages}{419--421}.
\newblock
\showISBNx{9781450362177}
\urldef\tempurl%
\url{https://doi.org/10.1145/3293611.3331575}
\showDOI{\tempurl}


\bibitem[Nikolaev and Ravindran(2020)]%
        {10.1145/3332466.3374540}
\bibfield{author}{\bibinfo{person}{Ruslan Nikolaev} {and}
  \bibinfo{person}{Binoy Ravindran}.} \bibinfo{year}{2020}\natexlab{}.
\newblock \showarticletitle{{Universal Wait-Free Memory Reclamation}}. In
  \bibinfo{booktitle}{\emph{Proceedings of the 25th ACM SIGPLAN Symposium on
  Principles and Practice of Parallel Programming}}. \bibinfo{publisher}{ACM},
  \bibinfo{address}{New York, NY, USA}, \bibinfo{pages}{130--143}.
\newblock
\showISBNx{9781450368186}
\urldef\tempurl%
\url{https://doi.org/10.1145/3332466.3374540}
\showDOI{\tempurl}


\bibitem[Nikolaev and Ravindran(2021a)]%
        {nikolaev_et_al:LIPIcs.DISC.2021.60}
\bibfield{author}{\bibinfo{person}{Ruslan Nikolaev} {and}
  \bibinfo{person}{Binoy Ravindran}.} \bibinfo{year}{2021}\natexlab{a}.
\newblock \showarticletitle{{Brief Announcement: Crystalline: Fast and Memory
  Efficient Wait-Free Reclamation}}. In \bibinfo{booktitle}{\emph{35th
  International Symposium on Distributed Computing (DISC 2021)}}
  \emph{(\bibinfo{series}{Leibniz International Proceedings in Informatics
  (LIPIcs)}, Vol.~\bibinfo{volume}{209})},
  \bibfield{editor}{\bibinfo{person}{Seth Gilbert}} (Ed.).
  \bibinfo{publisher}{Schloss Dagstuhl -- Leibniz-Zentrum f{\"u}r Informatik},
  \bibinfo{address}{Dagstuhl, Germany}, \bibinfo{pages}{60:1--60:4}.
\newblock
\showISBNx{978-3-95977-210-5}
\showISSN{1868-8969}
\urldef\tempurl%
\url{https://doi.org/10.4230/LIPIcs.DISC.2021.60}
\showDOI{\tempurl}


\bibitem[Nikolaev and Ravindran(2021b)]%
        {10.1145/3453483.3454090}
\bibfield{author}{\bibinfo{person}{Ruslan Nikolaev} {and}
  \bibinfo{person}{Binoy Ravindran}.} \bibinfo{year}{2021}\natexlab{b}.
\newblock \showarticletitle{{Snapshot-Free, Transparent, and Robust Memory
  Reclamation for Lock-Free Data Structures}}. In
  \bibinfo{booktitle}{\emph{Proceedings of the 42nd ACM SIGPLAN International
  Conference on Programming Language Design and Implementation}}.
  \bibinfo{publisher}{ACM}, \bibinfo{address}{New York, NY, USA},
  \bibinfo{pages}{987--1002}.
\newblock
\showISBNx{9781450383912}
\urldef\tempurl%
\url{https://doi.org/10.1145/3453483.3454090}
\showDOI{\tempurl}


\bibitem[Nikolaev and Ravindran(2022)]%
        {10.1145/3503221.3508440}
\bibfield{author}{\bibinfo{person}{Ruslan Nikolaev} {and}
  \bibinfo{person}{Binoy Ravindran}.} \bibinfo{year}{2022}\natexlab{}.
\newblock \showarticletitle{{POSTER: wCQ: A Fast Wait-Free Queue with Bounded
  Memory Usage}}. In \bibinfo{booktitle}{\emph{Proceedings of the 27th ACM
  SIGPLAN Symposium on Principles and Practice of Parallel Programming}}
  (Seoul, Republic of Korea) \emph{(\bibinfo{series}{PPoPP '22})}.
  \bibinfo{publisher}{ACM}, \bibinfo{address}{New York, NY, USA},
  \bibinfo{pages}{461--462}.
\newblock
\showISBNx{9781450392044}
\urldef\tempurl%
\url{https://doi.org/10.1145/3503221.3508440}
\showDOI{\tempurl}


\bibitem[Ramalhete and Correia(2016)]%
        {pedroWFQUEUEFULL}
\bibfield{author}{\bibinfo{person}{Pedro Ramalhete} {and}
  \bibinfo{person}{Andreia Correia}.} \bibinfo{year}{2016}\natexlab{}.
\newblock \showarticletitle{{A Wait-Free Queue with Wait-Free Memory
  Reclamation (Complete Paper)}}.
\newblock
\urldef\tempurl%
\url{https://github.com/pramalhe/ConcurrencyFreaks/blob/master/papers/crturnqueue-2016.pdf}
\showURL{%
\tempurl}


\bibitem[Ramalhete and Correia(2017a)]%
        {10.1145/3087556.3087588}
\bibfield{author}{\bibinfo{person}{Pedro Ramalhete} {and}
  \bibinfo{person}{Andreia Correia}.} \bibinfo{year}{2017}\natexlab{a}.
\newblock \showarticletitle{{Brief Announcement: Hazard Eras - Non-Blocking
  Memory Reclamation}}. In \bibinfo{booktitle}{\emph{Proceedings of the 29th
  ACM Symposium on Parallelism in Algorithms and Architectures}} (Washington,
  DC, USA) \emph{(\bibinfo{series}{SPAA '17})}. \bibinfo{publisher}{ACM},
  \bibinfo{address}{New York, NY, USA}, \bibinfo{pages}{367--369}.
\newblock
\showISBNx{9781450345934}
\urldef\tempurl%
\url{https://doi.org/10.1145/3087556.3087588}
\showDOI{\tempurl}


\bibitem[Ramalhete and Correia(2017b)]%
        {pedroWFQUEUE}
\bibfield{author}{\bibinfo{person}{Pedro Ramalhete} {and}
  \bibinfo{person}{Andreia Correia}.} \bibinfo{year}{2017}\natexlab{b}.
\newblock \showarticletitle{{POSTER: A Wait-Free Queue with Wait-Free Memory
  Reclamation}}. In \bibinfo{booktitle}{\emph{Proceedings of the 22nd ACM
  SIGPLAN Symposium on Principles and Practice of Parallel Programming}}
  (Austin, Texas, USA) \emph{(\bibinfo{series}{PPoPP '17})}.
  \bibinfo{publisher}{ACM}, \bibinfo{address}{New York, NY, USA},
  \bibinfo{pages}{453--454}.
\newblock
\showISBNx{978-1-4503-4493-7}
\urldef\tempurl%
\url{https://doi.org/10.1145/3018743.3019022}
\showDOI{\tempurl}


\bibitem[Sarkar et~al\mbox{.}(2012)]%
        {Sarkar:2012:SCP:2254064.2254102}
\bibfield{author}{\bibinfo{person}{Susmit Sarkar}, \bibinfo{person}{Kayvan
  Memarian}, \bibinfo{person}{Scott Owens}, \bibinfo{person}{Mark Batty},
  \bibinfo{person}{Peter Sewell}, \bibinfo{person}{Luc Maranget},
  \bibinfo{person}{Jade Alglave}, {and} \bibinfo{person}{Derek Williams}.}
  \bibinfo{year}{2012}\natexlab{}.
\newblock \showarticletitle{{Synchronising C/C++ and POWER}}. In
  \bibinfo{booktitle}{\emph{Proceedings of the 33rd ACM SIGPLAN Conference on
  Programming Language Design and Implementation}} (Beijing, China)
  \emph{(\bibinfo{series}{PLDI '12})}. \bibinfo{publisher}{ACM},
  \bibinfo{address}{New York, NY, USA}, \bibinfo{pages}{311--322}.
\newblock
\showISBNx{978-1-4503-1205-9}
\urldef\tempurl%
\url{https://doi.org/10.1145/2254064.2254102}
\showDOI{\tempurl}


\bibitem[{SPDK Developers}(2022)]%
        {SPDK}
\bibfield{author}{\bibinfo{person}{{SPDK Developers}}.}
  \bibinfo{year}{2022}\natexlab{}.
\newblock \showarticletitle{{Storage Performance Development Kit}}.
\newblock
\newblock
\shownote{\url{https://spdk.io/}}.


\bibitem[Tsigas and Zhang(2001)]%
        {Tsigas:2001:SFS:378580.378611}
\bibfield{author}{\bibinfo{person}{Philippas Tsigas} {and} \bibinfo{person}{Yi
  Zhang}.} \bibinfo{year}{2001}\natexlab{}.
\newblock \showarticletitle{{A Simple, Fast and Scalable Non-blocking
  Concurrent FIFO Queue for Shared Memory Multiprocessor Systems}}. In
  \bibinfo{booktitle}{\emph{Proceedings of the 13th ACM Symposium on Parallel
  Algorithms and Architectures}} (Crete Island, Greece)
  \emph{(\bibinfo{series}{SPAA '01})}. \bibinfo{pages}{134--143}.
\newblock
\showISBNx{1-58113-409-6}
\urldef\tempurl%
\url{https://doi.org/10.1145/378580.378611}
\showDOI{\tempurl}


\bibitem[Vyukov(2022)]%
        {vyakov}
\bibfield{author}{\bibinfo{person}{Dmitry Vyukov}.}
  \bibinfo{year}{2022}\natexlab{}.
\newblock \bibinfo{title}{{Bounded MPMC queue}}.
\newblock
\newblock
\newblock
\shownote{\url{http://www.1024cores.net/home/lock-free-algorithms/queues/bounded-mpmc-queue}}.


\bibitem[Wen et~al\mbox{.}(2018)]%
        {10.1145/3178487.3178488}
\bibfield{author}{\bibinfo{person}{Haosen Wen}, \bibinfo{person}{Joseph
  Izraelevitz}, \bibinfo{person}{Wentao Cai}, \bibinfo{person}{H.~Alan Beadle},
  {and} \bibinfo{person}{Michael~L. Scott}.} \bibinfo{year}{2018}\natexlab{}.
\newblock \showarticletitle{{Interval-Based Memory Reclamation}}. In
  \bibinfo{booktitle}{\emph{Proceedings of the 23rd ACM SIGPLAN Symposium on
  Principles and Practice of Parallel Programming}} (Vienna, Austria)
  \emph{(\bibinfo{series}{PPoPP '18})}. \bibinfo{publisher}{ACM},
  \bibinfo{address}{New York, NY, USA}, \bibinfo{pages}{1--13}.
\newblock
\showISBNx{9781450349826}
\urldef\tempurl%
\url{https://doi.org/10.1145/3178487.3178488}
\showDOI{\tempurl}


\bibitem[Yang and Mellor-Crummey(2016)]%
        {Yang:2016:WQF:2851141.2851168}
\bibfield{author}{\bibinfo{person}{Chaoran Yang} {and} \bibinfo{person}{John
  Mellor-Crummey}.} \bibinfo{year}{2016}\natexlab{}.
\newblock \showarticletitle{{A Wait-free Queue As Fast As Fetch-and-add}}. In
  \bibinfo{booktitle}{\emph{Proceedings of the 21st ACM SIGPLAN Symposium on
  Principles and Practice of Parallel Programming}} (Barcelona, Spain)
  \emph{(\bibinfo{series}{PPoPP '16})}. \bibinfo{publisher}{ACM},
  \bibinfo{address}{New York, NY, USA}, Article \bibinfo{articleno}{16},
  \bibinfo{numpages}{13}~pages.
\newblock
\showISBNx{978-1-4503-4092-2}
\urldef\tempurl%
\url{https://doi.org/10.1145/2851141.2851168}
\showDOI{\tempurl}


\end{thebibliography}

\appendix

\begin{figure*}
\begin{subfigure}{\columnwidth}
\begin{algorithm2e}[H]
\textbf{wCQ *} LHead = <empty \textbf{wCQ}>, LTail = LHead\;
\BlankLine
\Fn {\textbf{void} Enqueue\_Unbounded(\textbf{void *} p)} {
\textbf{wCQ *} ltail = hp.protectPtr(HPTail, Load(\&LTail))\;
\tcp{Enqueue\_Ptr() returns \textbf{false} if full or finalized}
\If {\upshape \Not ltail->next \AndOp ltail->Enqueue\_Ptr(p, finalize=True)} {
  hp.clear()\;
  \Return\;
}
\textbf{wCQ *} cq = alloc\_wCQ()\tcp*{Allocate wCQ}
cq->init\_wCQ(p)\tcp*{Initialize \& put p}
enqueuers[TID] = cq\tcp*{== Slow path (CRTurn) ==}
\For{\upshape i = 0 .. NUM\_THRDS-1}  {
  \If {\upshape enqueuers[TID] = \Null} {
    hp.clear()\;
    \Return\;
  }
  \textbf{wCQ *} ltail = hp.protectPtr(HPTail, Load(\&LTail))\;
  \lIf {\upshape ltail $\ne$ Load(\&LTail)} {\Continue}
  \If {\upshape enqueuers[ltail->enqTid] = ltail} {
    CAS(\&enqueuers[ltail->enqTid], ltail, \Null)\;
  }
  \For{\upshape j = 1 .. NUM\_THRDS}  {
    cq = enqueuers[(j + ltail->enqTid) \ModOp NUM\_THRDS]\;
    \lIf {\upshape cq = \Null} {\Continue}
    finalize\_wCQ(ltail)\tcp*{Duplicate finalize is OK since}
    CAS(\&ltail->next, \Null, cq)\tcp*{cq or another node follows}
    \Break\;
  }
  \textbf{wCQ *} lnext = Load(\&ltail->next)\;
  \If {\upshape lnext $\ne$ \Null} {
    finalize\_wCQ(ltail)\tcp*{Duplicate finalize is OK since}
    CAS(\&LTail, ltail, lnext)\tcp*{lnext or another node follows}
  }
}
enqueuers[TID] = \Null\;
hp.clear()\;
}

\Fn{\textbf{bool} dequeue\_rollback(\textbf{wCQ *} prReq, \textbf{wCQ *} myReq)} {
  deqself[TID] = prReq\;
  giveUp(myReq, TID)\;
  \If {\upshape deqhelp[TID] $\ne$ myReq} {
    deqself[TID] = myReq\;
    \Return False\;
  }
  hp.clear()\;
  \Return True\;
}
\end{algorithm2e}
\end{subfigure}%
\hspace{-1.5em}
\begin{subfigure}{\columnwidth}
\begin{algorithm2e}[H]
\setcounter{AlgoLine}{37}
\Fn{\textbf{void} finalize\_wCQ(\textbf{wCQ *} cq)} {
OR(\&cq->Tail, \{ .Value=0, .Finalize=1 \})\;
}

\Fn{\textbf{void *} Dequeue\_Unbounded()} {
\textbf{wCQ *} lhead = hp.protectPtr(HPHead, Load(\&LHead))\;
\tcp{skip\_last modifies the default behavior for Dequeue on}
\tcp{the last element in aq, as described in the text.}
\textbf{void *} p = lhead->Dequeue\_Ptr(skip\_last=True)\;
\If {\upshape p $\ne$ \Last} {
\If {\upshape p $\ne$ \Null $ $ \OrOp lhead->next = \Null} {
  hp.clear()\;
  \Return p\;
}
}
\textbf{wCQ *} prReq = deqself[TID]\tcp*{== Slow path (CRTurn) ==}
\textbf{wCQ *} myReq = deqhelp[TID]\;
deqself[TID] = myReq\;
\For{\upshape i = 0 .. NUM\_THRDS-1}  {
  \lIf {\upshape deqhelp[TID] != myReq} {\Break}
  \textbf{wCQ *} lhead = hp.protectPtr(HPHead, Load(\&LHead))\;
  \lIf {\upshape lhead $\ne$ Load(\&LHead)} {\Continue}
  \textbf{void *} p = lhead->Dequeue\_Ptr(skip\_last=True)\;
  \If {\upshape p $\ne$ \Last} {
  \If {\upshape p $\ne$ \Null $ $ \OrOp lhead->next = \Null} {
    \lIf {\upshape \Not dequeue\_rollback(prReq, myReq)} {\Break}
    \Return p\;
  }
  Store(\&lhead->aq.Threshold, $3n - 1$)\;
  p = lhead->Dequeue\_Ptr(skip\_last=True)\;
  \If {\upshape p $\ne$ \Last $ $ \AndOp p $\ne$ \Null} {
    \lIf {\upshape \Not dequeue\_rollback(prReq, myReq)} {\Break}
    \Return p\;
  }
  }
  \textbf{wCQ *} lnext = hp.protectPtr(HPNext, Load(\&lhead->next))\;
  \lIf {\upshape lhead $\ne$ Load(\&LHead)} {\Continue}
  \lIf {\upshape searchNext(lhead, lnext) $\ne$ NOIDX} {casDeqAndHead(lhead, lnext)}
}
\textbf{wCQ *} myCQ = deqhelp[TID]\;
\textbf{wCQ *} lhead = hp.protectPtr(HPHead, Load(\&LHead))\;
\If {\upshape lhead = Load(\&LHead) \AndOp myCQ = Load(\&lhead->next)} {
  CAS(\&LHead, lhead, myCQ)\;
}
hp.clear()\;
hp.retire(prReq)\;
\Return myCQ->Locate\_Last\_Ptr()\;
}
\end{algorithm2e}
\end{subfigure}%
\caption{Adapting CRTurn to an unbounded wCQ-based queue design (high-level methods).}
\label{alg:lwfcq}
\end{figure*}

\section{Appendix: Unbounded Queue}
\label{sec:unbounded}

LSCQ and LCRQ implement unbounded queues by using an outer
layer of M\&S lock-free queue which links ring buffers together.
Since operations on the outer layer are very rare, the cost
is dominated by ring buffer operations. wCQ can follow the same idea.

Although the outer layer does not have to be performant, it still must be wait-free with bounded memory usage. However, M\&S queue is only lock-free. The (non-performant) CRTurn wait-free queue~\cite{pedroWFQUEUEFULL, pedroWFQUEUE} does satisfy the aforementioned requirements. Moreover, CRTurn already implements wait-free memory reclamation by using hazard pointers in a special way. wCQ and CRTurn combined together would yield a fast queue with bounded memory usage.

Because CRTurn's design is non-trivial and is completely orthogonal to the wCQ presentation, its discussion is beyond the scope of this paper. We refer the reader to~\cite{pedroWFQUEUEFULL, pedroWFQUEUE} for more details about CRTurn. Below, we sketch expected changes to CRTurn (assuming prior knowledge of CRTurn) to link ring buffers rather than individual nodes.
In Figure~\ref{alg:lwfcq}, we present pseudocode (assuming that entries are pointers) with corresponding changes to enqueue and dequeue operations in CRTurn.
For convenience, we retain the same variable and function names as in~\cite{pedroWFQUEUEFULL} (e.g., \textit{giveUp} that is not shown here). Similar to~\cite{pedroWFQUEUEFULL}, we assume memory reclamation API based on hazard pointers (the \textit{hp} object).

The high-level idea is that we create a wait-free queue (list) of
wait-free
ring buffers, where \textit{LHead} and \textit{LTail} represent corresponding
head and tail of the list. \textit{Enqueue\_Unbounded} will first attempt
to insert an entry to the last ring buffers as long as \textit{LTail} is
already pointing to the last ring buffer. Otherwise, it allocates a new ring buffer and inserts a new element. It then follows CRTurn's procedure to insert
the ring buffer to the list. The only difference is that when helping to
insert the new ring buffer, threads will make sure that the previous ring buffer is finalized.

\textit{Dequeue\_Unbounded} will first attempt to fetch an element from the
first ring buffer. wCQ's Dequeue for \textbf{aq} needs to be modified to
detect the very last entry in a \textit{finalized} ring buffer. (Note that it can only be done for finalized ring buffers, where no subsequent entries can be inserted.) Instead of returning the true entry, \textit{Dequeue\_Ptr} returns a special
\textbf{last} value. This approach helps to retain CRTurn's wait-freedom
properties as every single ring buffer contains at least one entry.
Helper methods must also be modified accordingly.

\end{document}